\documentclass[english,twocolumn,prx,amssymb,amsmath,superscriptaddress]{revtex4}

\usepackage{amsfonts,amsthm,mathrsfs,eufrak,xspace,framed}
\usepackage{color}
\usepackage{algorithm}
\floatname{algorithm}{Protocol}
\usepackage{algorithmic}
\usepackage{subfigure}
\usepackage{amscd}
\usepackage{multirow}
\usepackage{graphicx}
\usepackage{dcolumn}
\usepackage{bm}
\usepackage{bbm}
 \usepackage{enumitem}

\newtheorem{theorem}{Theorem}

\newcommand{\bra}[1]{\mbox{$\left\langle #1 \right|$}}
\newcommand{\ket}[1]{\mbox{$\left| #1 \right\rangle$}}

\newcounter{protocol}

\newcommand{\be}{\begin{equation}}
\newcommand{\ee}{\end{equation}}
\newcommand{\bea}{\begin{eqnarray}}
\newcommand{\eea}{\end{eqnarray}}

\definecolor{mygreen}{rgb}{0,0.5,0}
\definecolor{myblue}{rgb}{0,0,0.75}
\definecolor{mymagenta}{cmyk}{0,1,0,0.12}

\usepackage{umoline}

\begin{document}
\title{Quantum Private Function Evaluation}
\author{Zhu Cao}
\email{caozhu@ecust.edu.cn}
\affiliation{Key Laboratory of Smart Manufacturing in Energy Chemical Process, Ministry of Education, East China University of Science and Technology, Shanghai 200237, China}

\begin{abstract}
Private function evaluation is a task that aims to obtain the output of a function while keeping the function secret. 
So far its quantum analogue has not yet been articulated.
In this study, we initiate the study of quantum private function evaluation, the quantum analogue of classical 
private function evaluation. We give a formal definition of quantum private function evaluation and present two schemes together with their security proofs.  
We then give an experimental demonstration of the scheme. Finally we apply quantum private function evaluation to quantum copy protection to illustrate its usage.
\end{abstract}

\maketitle

\section{Introduction}
\label{sec:introduction}



Consider the following scenario. A software company wishes to provide a charged quantum software service on the cloud to the customers,
but it does not wish to reveal the code of the quantum software as the users will just download the code and not pay any money for the 
web service. The customer, on the other hand, wishes to apply the software on their quantum data for some information processing, but 
does not wish to reveal their private data to the software company as the data may contain sensitive information. For the classical
version of the scenario, it can be dealt with by classical private function evaluation, which 
has been widely studied since 1990 \cite{abadi1990secure,mohassel2014actively}.


Quantum private function evaluation (QPFE) is a quantum version of private function evaluation (PFE). 
While private function evaluation has been extensively studied in classical cryptography, somewhat surprisingly, up till now,
QPFE has not yet been studied. 
The setting of QPFE is as follows. Party 1 has the quantum function $\mathcal{E}$ and part of the quantum state $\rho_1$.
Party 2 has the other part of the quantum state $\rho_2$. (The case of $N$ parties is similar. Here we explain the case of two parties for simplicity.)

The goal of the two parties is to obtain 
\begin{equation}
\mathcal{E}(\rho_1, \rho_2) = (\mathcal{E}^{(1)}(\rho_1, \rho_2), \mathcal{E}^{(2)}(\rho_1, \rho_2))  \triangleq (\tau_1, \tau_2)
\end{equation}
without revealing each party's private quantum inputs.
Here, $\mathcal{E}^{(i)}(\rho_1, \rho_2)=\tau_i$ is the quantum output to party $i$.


The protocol needs to satisfy the following correctness and privacy properties:
\begin{enumerate}
\item Correctness: the malicious parties $\mathcal{A}$ cannot influence the results of the honest parties other than changing $\mathcal{A}$'s own inputs.
\item Privacy: The privacy contains two parts. For privacy against malicious party 2, we require party 2 to obtain 
no information on party 1's inputs and functions beyond $\mathcal{E}^{(2)}(\rho_1, \rho_2)$. For party 1, since it can change the function, privacy against 
malicious party 1 can only be achieved when party 1 has no function output. Otherwise, we can only
consider privacy against semi-honest party 1, namely party 1 that does exactly as the protocol dictates, but is 
curious about the other party's private input. 
\end{enumerate}
A formal formulation of the above statements will be presented later in the paper.


In this work, we initiate the study of QPFE. We first define QPFE.
Then we present two QPFE schemes together with their security proofs. The first scheme is more modular in the sense that it 
uses  multi-party quantum computation (MPQC) as a black box. Usually modularity makes schemes simpler and less error-prone in real implementations.
This scheme only relies on the same assumption as MPQC. (Currently it is known that a quantum one-way function is sufficient for MPQC  \cite{bartusek2021one}.)
 The second scheme does not use MPQC as a black box, and hence has the potential of being more efficient.
 We then give an experimental demonstration of QPFE based on the second scheme. To show the usefulness 
 of QPFE, we also apply the QPFE scheme to the task of quantum copy protection. We hope our work would stimulate further 
 research on developing alternative QPFE schemes and finding use of QPFE in other applications.


The organization of the rest of the paper is as follows. 
In Sec.~\ref{sec:relatedworks}, we review some related works.
In Sec.~\ref{sec:preli}, we fix the notation and restate some results from the literature that will be used later in the paper.
In Sec.~\ref{sec:def}, we present the definition of QPFE. 
In Sec.~\ref{sec:protocol}, we present two protocols of QPFE together with some extensions.
In Sec.~\ref{sec:securityproof}, we present the security proofs of the QPFE protocols. 
In Sec.~\ref{sec:exp}, we give an experimental demonstration of QPFE.
In Sec.~\ref{sec:appl}, we give an example application of QPFE to quantum copy protection.
Finally, in Sec.~\ref{sec:conclusion}, we conclude the paper and give a few outlooks. 

\section{Related work}
\label{sec:relatedworks}


\subsection{Classical private function evaluation }

Classical PFE has been widely studied since 1990 \cite{abadi1990secure}. 
Currently there are roughly two classes of classical PFE. One class of classical PFE is based
on universal circuits and multi-party computation (MPC) \cite{kiss2016valiant,gunther2017more}.
This line of work focuses on optimizing the implementation of universal circuits to achieve higher
efficiency of PFE. Once the universal circuit part is done, one puts the private function also as 
part of the input into the universal circuit and then applies ordinary MPC. The security proof 
for this type of PFE is almost trivial which is one of the benefits of such schemes.

The other class of classical PFE does not rely on 
 universal circuits \cite{mohassel2014actively,bingol2019efficient,biccer2020highly}.
 Compared to the previous type of PFE, this type of PFE can achieve higher efficiency by
 opening the black box of MPC. We refer the reader to the related work section of Ref.~\cite{biccer2020highly}
 for a recent review of this line of PFE. We only add the remark here that these classical PFEs 
 usually involve copying the output of a gate several times as the inputs of further gates. 
 This makes it unclear how to transfer these classical PFE schemes to the quantum scenario.

\subsection{Multi-party quantum computation}

PFE is related to MPC. Both tasks consider the task of computing a joint function
with multiple parties. Both tasks wish to protect the privacy of the participants' data. There are however also essential 
differences between PFE and MPC. The function to be computed is private information in PFE, while it is public in MPC. 
The quantum analogue of MPC, called MPQC, has been studied since 2002 \cite{crepeau2002secure}. 

Since MPQC is closely related to quantum PFE, we give a brief review of MPQC here. 
A more comprehensive review of the field of MPQC can be found in Ref.~\cite{bartusek2021round}. 
In the first work of MPQC \cite{crepeau2002secure}, Crepeau {\it et al.} present an
information-theoretically secure MPQC protocol when the proportion of malicious parties is less than $1/6$.
Since then, the results have been strengthened in various directions. In 2006, Ben-Or {\it et al.} \cite{ben2006secure}     
showed that MPQC can be done information theoretically with a strict honest majority. This work is tight in the sense that, even with one more malicious party, MPQC will be completely impossible information theoretically. In 2020, Dulek {\it et al.} \cite{dulek2020secure} showed that with computational assumptions, MPQC can be done with a dishonest majority (even if just one party is honest).
The computation assumption needed for MPQC is that a quantum one-way function exists \cite{bartusek2021one}. 
There are also works concerning the efficiency of MPQC. For example, in 2021, a constant-round MPQC is given for the case of a dishonest majority \cite{bartusek2021round}, reducing the round complexity of previous works.

\subsection{Quantum fully homomorphic encryption}

Quantum fully homomorphic encryption (QFHE) considers the following scenario. One party, called the user, has private quantum data $x$
and wishes to get the result of $F(x)$ but it does not have the computational resource to do so. Instead, it 
asks the server to do the computation but does not wish to reveal $x$ to the server. The user encrypts the data $x$ 
as $\mathcal{E}(x)$ and sends it to the server. The server computes homomorphically on  $\mathcal{E}(x)$ to obtain $\mathcal{E}(F(x))$
which is sent back to the user. The user decrypts $\mathcal{E}(F(x))$ as $F(x)$. A notable distinction between QFHE and blind quantum
computation is that QFHE does not allow interactions between the user and the server in the middle. The only interaction allowed 
is the message that the user sends to the server at the beginning, and the message that the server sends to the user at the end.

It has been shown that quantum FHE is not possible for general functions \cite{yu2014limitations}. Hence, the research on QFHE 
is mainly divided into two classes. The first class of research considers adding computational assumptions for achieving QFHE.
For example, in 2015, Broadbent {\it et al.} \cite{broadbent2015quantum} proposed a QFHE scheme with a constant number of T gates based 
on computational assumptions. Later, in 2016, the restriction on the number of T gates is removed by Dulek {\it et al.} \cite{dulek2016quantum}.
Several further improvements are made \cite{alagic2017quantum,mahadev2018classical2,brakerski2018quantum}. 
The second class of research considers QFHE with special functions instead of general functions. For example, Ouyang {\it et al.} 
showed an informational-theoretically secure QFHE scheme with a constant number of T gates in 2018 \cite{ouyang2018quantum}.

QFHE \cite{dulek2016quantum} can be used to fulfill the goal of QPFE, but suffers from the problems of high complexity and 
strong assumptions. All known protocols of QFHE for general functions \cite{dulek2016quantum,alagic2017quantum,mahadev2018classical2,brakerski2018quantum}
rely on the assumption that learning with errors is quantum hard. Later in our paper, we will show that 
for QPFE with a dishonest majority, the existence of a quantum one-way function is enough, just as the case for MPQC \cite{bartusek2021one}.

 \section{Preliminaries} 
\label{sec:preli}

In this section, we present some definitions and results from the literature which will be used later in the paper.
 In the first subsection, we review a security framework that will be later used to define the security of QPFE. The remaining three subsections
describe three primitives that will be later used to construct QPFE protocols.

\subsection{General secure protocol}
We begin by reviewing how to define the security of a general secure protocol $\Pi$.
Before presenting the security definition, we first define a few notations.
Let $\lambda$ be the security parameter. For two classes of binary random variables
$\{X_\lambda\}_\lambda$ and $\{Y_\lambda\}_\lambda$, we write $X_\lambda \approx Y_\lambda$
as a shorthand for
\begin{equation}
| \textrm{Pr}(X_\lambda=1)  - \textrm{Pr}(Y_\lambda=1)  |  \le \mu(\lambda), \quad \forall \lambda,
\end{equation}
where $\textrm{Pr}$ stands for probability and $\mu$ is a negligible function of $\lambda$.  Here, 
a negligible function of $\lambda$ is defined as a positive function which tends to 0 faster than the inverse of any positive polynomial function
as $\lambda$ goes to infinity.
An example of a negligible function of $\lambda$ is $2^{-\lambda}$, which satisfies this definition.

We employ the simulation-based security proof framework \cite{canetti2001universally} to define the security of a general protocol. 
In this framework, a protocol is considered secure if any environment with a bounded amount of computational resources cannot distinguish between the actual protocol execution  (usually called the real world)  and an ideal protocol execution (usually called the ideal world) with a high probability. In an ideal world, there is a trusted third-party that fulfills all security properties that the protocol should satisfy, which is usually called the ideal functionality. To ensure the environment cannot distinguish the two cases, one plausible way is to add a simulator wrapping the adversary such that the interaction of the environment with the simulator and the ideal functionality in the ideal world is identical to the interaction of the environment with the adversary and the protocol in the real world.

In more details, let $\textrm{Exec}[\Pi, \mathcal{A}, \mathcal{Z}]$ be the output distribution of the environment $\mathcal{Z}$ for a protocol
$\Pi$ and an adversary $\mathcal{A}$ compatible with the protocol $\Pi$. The relations
among $\mathcal{Z}$, $\Pi$, and $\mathcal{A}$ are shown in Fig.~\ref{fig:RvsI}(a). They are capable of communicating 
with each other in the real world.
Let $\textrm{Exec}[\mathcal{F}, \mathcal{S}, \mathcal{Z}]$ be the output distribution of the environment $\mathcal{Z}$ for
an ideal functionality $\mathcal{F}$ and a simulator $S$. The relations
among $\mathcal{Z}$, $\mathcal{F}$, and $\mathcal{S}$ are shown in Fig.~\ref{fig:RvsI}(b). They are capable of communicating 
with each other in the ideal world.

\begin{figure}[htb]
\centering \includegraphics[width=6cm]{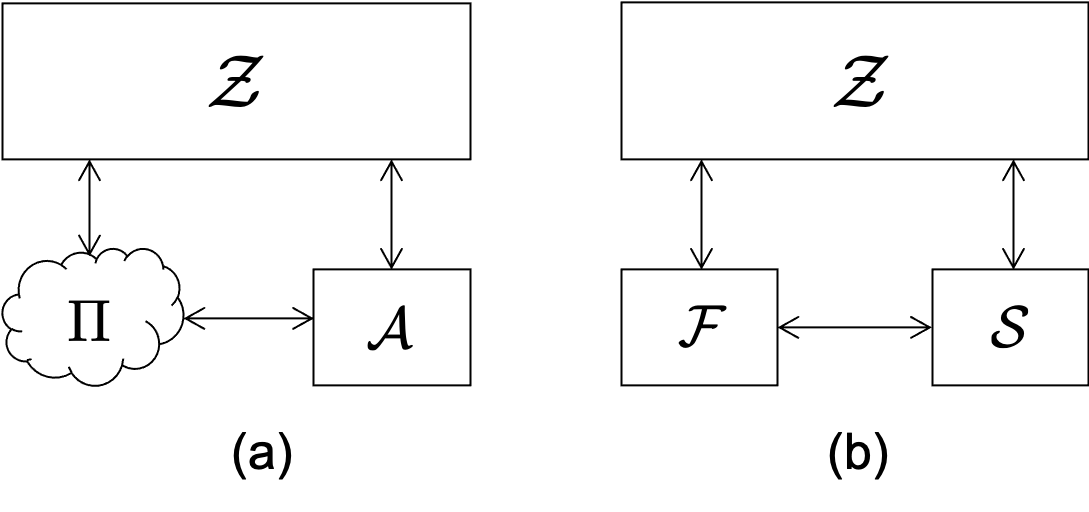}
\caption{(a) The real world. (b) The ideal world.}
\label{fig:RvsI}
\end{figure}

We say that a protocol $\Pi$ realizes an ideal functionality $\mathcal{F}$ if 
for all adversary $\mathcal{A}$, there exists a simulator $\mathcal{S}$ that may depend on $\mathcal{A}$ 
such that 
\begin{equation}
\textrm{Exec}[\Pi, \mathcal{A}, \mathcal{Z}] \approx \textrm{Exec}[\mathcal{F}, \mathcal{S}, \mathcal{Z}].
\end{equation}
 When the output distributions are non-binary, the notation $\approx$ can be naturally extended to imply that the total variational distance between 
the two distributions is a negligible function of the security parameter $\lambda$. 
With this security definition, we only need to specify the ideal functionality to define any secure protocol.

\subsection{Two-party computation}

The syntax of the two-party computation (2PC) functionality is as follows \cite{Ishai2008founding}. 
Initially, Bob holds a private input $x_B$ and Alice holds a private input $x_A$. They aim to compute the output $(y_A, y_B) = F(x_A, x_B)$ without revealing their private inputs to each other.
Here, $y_A$ is the output for Alice and $y_B$ is the output for Bob.

There are three rounds in the 2PC functionality with the following syntax. 
\begin{enumerate}
\item Bob generates $(m_B, st) = \textsf{2PC}_1( x_B)$ and sends $m_B$ to Alice.
\item Alice sends $m_A = \textsf{2PC}_2(x_A, m_B)$ to Bob.
\item Bob computes $(y_A, y_B) = \textsf{2PC}_{out}(m_A, st)$ and sends $y_A$ to Alice. 
\end{enumerate}


There are both simulators for Alice and Bob.
Let us start from Alice's side. The simulator for Alice has the form  $Sim_A = ( Sim_A^{(1)}, Sim_A^{(2)} ) $.  
The simulator uses $Sim_A^{(1)}$ to generate $m_B$ and sends it to Alice.
On receiving $m_A$, the simulator  calculates $Sim_A^{(2)}(m_A) = x_A' $ and sends $x_A'$ to the ideal functionality $I[x_B](\cdot)$.
The ideal functionality calculates $y = F(x_A', x_B)$ and sends it to the simulator. The simulator then sends $y$ to Alice.
The random variable $IDEAL_A$ consists of Alice's view in the ideal world,
and the random variable $REAL_A$ consists of Alice's view in the real world. The security against malicious Alice holds if
\begin{equation}
IDEAL_A \approx REAL_A.
\end{equation}

The simulator for Bob has the form $Sim_B = ( Sim_B^{(1)}, Sim_B^{(2)}) $.  
On receiving $m_B$ from Bob, $Sim_B^{(1)}(m_B)$ generates $x_B'$ and sends it to the ideal functionality $I[x_A](\cdot)$.
The ideal functionality generates $y \leftarrow I[x_A](x_B')$ and gives it to the simulator.
On receiving $y$, $Sim_B^{(2)}(y)$ generates $m_A$ and gives it to Bob.
The random variable $IDEAL_B$ consists of Bob's view in the ideal world, and 
 the random variable $REAL_B$ consists of Bob's view in the real world. The security against malicious Bob holds if
\begin{equation}
IDEAL_B \approx REAL_B.
\end{equation}

\subsection{Quantum garbled circuit}

Here, we consider only the quantum garbled circuit for Clifford $+$ Measurement ($C+M$) circuits, which are circuits that only consist of Clifford gates and 
standard-basis measurements. 

The syntax of the quantum garbled circuit consists of a triad of functions ($\textsf{QGarble}$, $\textsf{QGEval}$, $\textsf{QGSim}$).
The $C+M$ quantum circuit is denoted as $Q$. The quantum input of the circuit is denoted as $x_{inp}$.
The $C+M$ quantum circuit consists of $d$ layers and initially contains $n_0: = n_1+ k_1$ qubits. In layer $i$ ($1\le i\le d$), first a Clifford is operated on the $n_{i-1}$ qubits 
and then the first $k_i$ qubits are measured. The number of remaining qubits is denoted as $n_i$ and satisfies $n_i = n_{i-1}-k_i$.

The three functions have the following syntax:
\begin{enumerate}
\item $(E_0, \tilde{Q}) \leftarrow \textsf{QGarble}(Q, \{ n_i, k_i\}, 1^\lambda)$.
Here, $\textsf{QGarble}$ is the function that turns a quantum circuit into its garbled version, $E_0$ is the unitary operation to be applied to the quantum input, $\lambda$ is the security parameter, and $\tilde{Q}$ is the garbled version of the quantum circuit $Q$.
\item  $\widetilde{x_{inp}} = E_0 (x_{inp}, 0^\lambda)$. Here, $0^\lambda$ is short for $\ket{0}^{\otimes \lambda}$, $x_{inp}$ contains $n_0$ qubits, $E_0$ is a unitary acting on $n_0+\lambda$ qubits, and $\widetilde{x_{inp}}$ is the encrypted input.
\item $x_{out} \leftarrow \textsf{QGEval}(\widetilde{x_{inp}}, \tilde{Q})$. Here, $\textsf{QGEval}$ is the evaluation function of the garbled circuit, and $x_{out}$ is the output after evaluating the garbled circuit $\tilde{Q}$ on the encrypted input $\widetilde{x_{inp}}$.
\item $(\widetilde{x_{inp}}, \tilde{Q}) \leftarrow \textsf{QGSim}(x_{out}, \{n_i, k_i\},1^\lambda) $. Here, $\textsf{QGSim}$ is the simulation function of the garbled circuit.
\end{enumerate}

\textbf{Correctness:} Essentially, the correctness holds when the output of the evaluation of the garbled circuit is identical to the output when running the original circuit directly on the original inputs.
This is equivalent to saying that
\begin{equation}
\{ \textsf{QGEval}(E_0(x_{inp} ), \tilde{Q}) \} \approx_s \{ Q(x_{inp}) \},
\end{equation}
where $\approx_s$ denotes statistical indistinguishability. 
Here, two density matrices 
$x_\lambda$ and $y_\lambda$ with the security parameter $\lambda$ are said to be statistically indistinguishable $x_\lambda \approx_s y_\lambda$ if
they satisfy
\begin{equation}
 || x_\lambda  - y_\lambda ||_1  \le \mu(\lambda),
\end{equation}
where $|| \cdot  ||_1$ is the trace distance and $\mu(\cdot)$ is a negligible function.

\textbf{Security:} Essentially, the security will be guaranteed if the garbled circuit and the garbled input can be reproduced by the output of the circuit, thus 
revealing no extra information on each party's private inputs. This is equivalent to saying that
\begin{equation}
 \{  (E_0(x_{inp}), \tilde{Q})   \} \approx_c \{ \textsf{QGSim}(Q(x_{inp}), \{n_i, k_i\}, 1^\lambda) \},
\end{equation}
where $\approx_c$ denotes computational indistinguishability. 
Here, two density matrices $x_\lambda$ and $y_\lambda$ with the security parameter $\lambda$ are said to be computionally indistinguishable $x_\lambda \approx_c y_\lambda$ if
for any probabilistic polynomial-time distinguisher $D_\lambda$ and any auxiliary quantum advice $d_\lambda$, the inequality
\begin{equation}
 | \textrm{Pr}[ D_\lambda(d_\lambda, x_\lambda) = 1 ] - \textrm{Pr}[ D_\lambda(d_\lambda, y_\lambda) = 1 ] |  \le \mu(\lambda)
\end{equation}
always holds for some negligible function $\mu(\cdot)$.  

The construction of a quantum garbled circuit first appears in Ref.~\cite{brakerski2022quantum}, and is then rigorously proved in Ref.~\cite{bartusek2021round}.

\subsection{Two-party quantum computation}

The syntax of two-party quantum computation (2PQC) is as follows. 
Initially, Bob holds a private quantum input $x_B$ and Alice holds a private quantum input $x_A$. 
They aim to compute the quantum output $(y_A, y_B) = F(x_A, x_B)$ without revealing their private quantum inputs to each other.

The 2PQC protocol consists of three functions $\textsf{2PQC}_1,\textsf{2PQC}_2,\textsf{2PQC}_{out}$ and three rounds as follows: 
\begin{enumerate}
\item  Bob generates $(m_B, st) = \textsf{2PQC}_1( x_B)$ and sends a quantum message $m_B$ to Alice. Here $st$ is a quantum state that Bob keeps for himself for later use.
\item Alice sends a quantum message $m_A = \textsf{2PQC}_2(x_A, m_B)$ to Bob.
\item  Bob computes $(y_A', y_B') = \textsf{2PQC}_{out}(m_A, st)$ and sends $y_A'$ to Alice. 
\end{enumerate}

\textbf{Correctness:} The correctness means that the output of the protocol should be identical to the output of
just computing the function $F$ on $x_A$ and $x_B$ directly. This is equivalent to saying that
\begin{equation}
 (y_A, y_B) \approx_s  (y_A', y_B').
\end{equation}

\textbf{Security:} The security consists of two parts, security against malicious Alice, and security against malicious Bob.
Let us start from Alice's side. The simulator for Alice has the form  $Sim_A = ( Sim_A^{(1)}, Sim_A^{(2)} ) $.  
The simulator uses $Sim_A^{(1)}$ to generate $m_B$ and sends it to Alice.
On receiving $m_A$, the simulator  calculates $Sim_A^{(2)}(m_A) = x_A' $ and sends $x_A'$ to the ideal functionality $I[x_B](\cdot)$.
The ideal functionality calculates $y_A' = F_A(x_A', x_B)$ and sends it to the simulator. 
The simulator then sends $y_A'$ to Alice.
Here $F_A$ is almost identical to $F$
except it only outputs the output for Alice.
The random variable $IDEAL_A$ consists of Alice's view, and let $REAL_A$ be the transcript 
of Alice in the real execution. The security is satisfied if
\begin{equation}
IDEAL_A \approx_c REAL_A.
\end{equation}

The simulator for Bob has the form $Sim_B = ( Sim_B^{(1)}, Sim_B^{(2)}) $.  
On receiving $m_B$ from Bob, $Sim_B^{(1)}(m_B)$ generates $x_B'$ and sends it to the ideal functionality $I[x_A](\cdot)$.
The ideal functionality generates $y \leftarrow I[x_A](x_B')$ and gives it to the simulator.
On receiving $y$, $Sim_B^{(2)}(y)$ generates $m_A$ and gives it to  Bob.
The random variable $IDEAL_B$ consists of Bob's view. Let the random variable $REAL_B$ be the view of Bob
in the real execution of the protocol. The security against malicious Bob holds if and only if
\begin{equation}
IDEAL_B \approx_c REAL_B.
\end{equation}

\section{Definition of QPFE}
\label{sec:def}

One of our first contributions is to formally define the task of QPFE and the security properties it should satisfy.
Let us start by describing its setup. In a QPFE, there are $n$ ($n\ge 2$) parties. Party 1 holds a private function $\mathcal{E}$ and a private quantum input $\rho_1$.
Party $i$ ($2\le i \le n$) holds a private quantum input $\rho_i$. Note that $\rho_1,\cdots, \rho_n$ can be entangled. 
An illustration of QPFE is shown in Fig.~\ref{fig:illlus}(a). After a few communication and computation rounds, party $i$ gets a quantum output $\sigma_i$.
Here, $\sigma_1, \cdots, \sigma_n$ can also be entangled.

\begin{figure}[htb]
\centering \includegraphics[width=8cm]{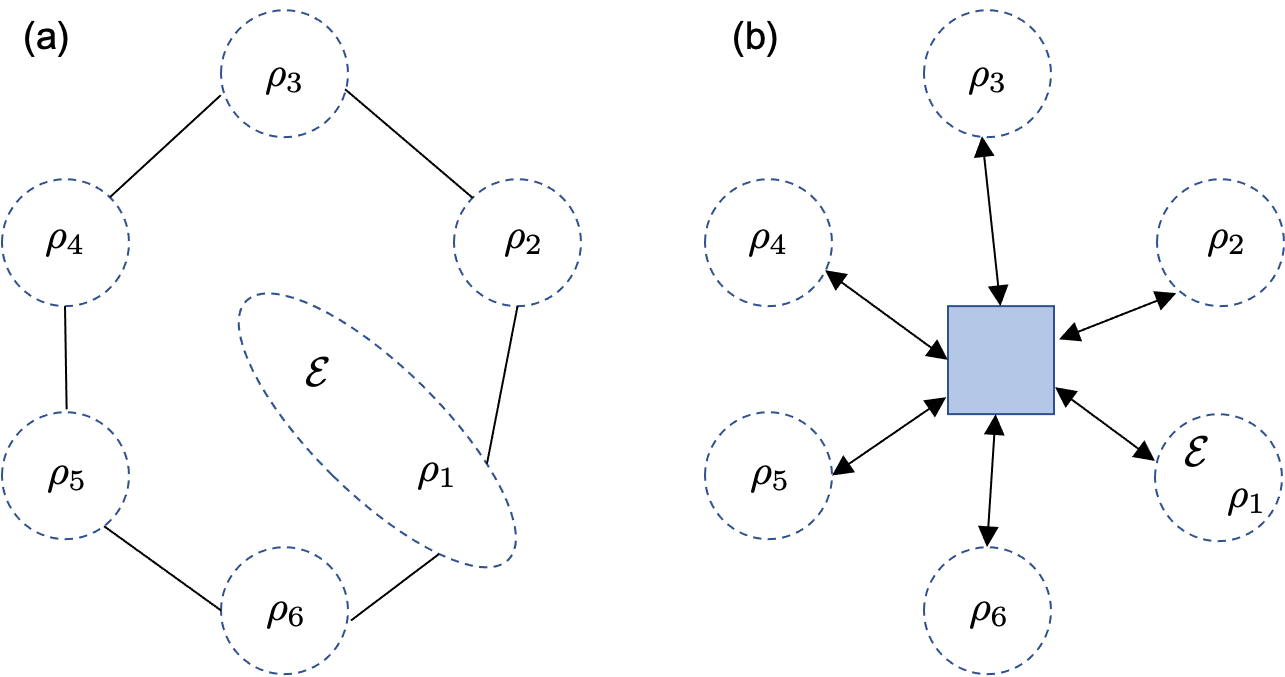}
\caption{(a) The real scenario where mutually untrusted parties communicate with each other. (b) The ideal scenario where each party interacts with a trusted third-party.}
\label{fig:illlus}
\end{figure}

In the introduction, we have informally described the security properties that a QPFE should satisfy: 
correctness and privacy. To give a formal definition, we adopt the simulation-based security framework.
For the convenience of the readers who are not familiar with this framework, we briefly review it here.
In this framework, there is a real world and an ideal world. The real world is where the protocol actually
takes place and there is an adversary that tries to attack the protocol. The ideal world, on the other hand, contains an ideal functionality that fulfills the task at 
hand (e.g., QPFE), and there is a simulator that interacts with the ideal functionality.
Now say that the ideal functionality fulfills the security requirements of correctness and privacy against any adversary by its construction, and there
is no way to tell the real world and the ideal world apart, then the adversary in the real world also cannot break the security of the protocol.

%
%
%

The above reasoning is quite high-level. We now give more details of this framework in the specific case
of QPFE. As explained previously, the ideal functionality $\mathcal{F}_{QPFE}$ should fulfill the conditions 
of correctness and privacy against any adversary. To model possible collusion between dishonest parties,
we assume there is a single adversary controlling all dishonest parties. In other words, we can consider 
the collection of all dishonest parties as the adversary. After elucidating the adversary, we are now ready 
to provide further details of the ideal functionality. First of all, the ideal functionality $\mathcal{F}_{QPFE}$
has the following input and output:
\begin{enumerate}
\item Input: The ideal functionality takes input from both honest parties and the adversary. Note that the adversary sends 
the inputs of dishonest parties on behalf of the dishonest parties. As said before, Party 1 provides inputs $\mathcal{E}$ and  $\rho_1$ to the ideal functionality,
and Party $i$ provides inputs $\rho_i$ to the ideal functionality where $2 \le i \le n$.
\item Output: The ideal functionality gives outputs to both honest parties and dishonest parties. As said before, 
 Party $i$ receives $\mathcal{E}^{(i)}(\rho_1, \cdots, \rho_n)$ from the ideal functionality.
\end{enumerate}

To satisfy the correctness and privacy properties, the ideal functionality works as follows.
\begin{enumerate}
\item After receiving the inputs from all parties, compute $\mathcal{E}(\rho_1,\cdots,\rho_n) =( \mathcal{E}^{(1)}(\rho_1,\cdots,\rho_n),\cdots, \mathcal{E}^{(n)}(\rho_1,\cdots,\rho_n) )$.
\item The ideal functionality receives a message ``abort'' or ``continue'' from the adversary. If ``abort'' is received, the process halts. 
If ``continue''  is received, the ideal functionality sends $\mathcal{E}^{(i)}(\rho_1,\cdots,\rho_n)$  to Party $i$.
\end{enumerate}

By the definition of this ideal functionality, it is clear that it always computes the outputs correctly and it does not leak any private 
information to any party. 
An illustration of this ideal functionality is shown in Fig.~\ref{fig:illlus}(b).
Note that this ideal functionality does not provide guaranteed output delivery, as the 
honest parties will not get any output if the adversary sends the message ``abort'' in the second step.  We can
modify the ideal functionality into an ideal functionality with guaranteed output delivery by removing the ability of the adversary to halt the process.

Using this ideal functionality, we are now ready to define the security property of QPFE. 
We say that a protocol $\Pi$ realizes QPFE if 
for all polynomial-time adversary $\mathcal{A}$, there exists a polynomial-time simulator $\mathcal{S}$ that may depend on $\mathcal{A}$ 
such that 
\begin{equation}
\textrm{Exec}[\Pi, \mathcal{A}, \mathcal{Z}] \approx_c \textrm{Exec}[\mathcal{F}_{QPFE}, \mathcal{S}, \mathcal{Z}].
\end{equation}
More precisely, 
let $\lambda$ be the security parameter.
If
for any probabilistic polynomial-time distinguisher $D_\lambda$ and any auxiliary quantum advice $d_\lambda$, the inequality
\begin{equation}
\begin{aligned}
& | \textrm{Pr}[ D_\lambda(d_\lambda, \textrm{Exec}[\Pi, \mathcal{A}, \mathcal{Z}]) = 1 ] -  \\
 & \quad \textrm{Pr}[ D_\lambda(d_\lambda, \textrm{Exec}[\mathcal{F}_{QPFE}, \mathcal{S}, \mathcal{Z}]) = 1 ] |  \le \mu(\lambda)
\end{aligned}
\end{equation}
always holds for some negligible function $\mu(\cdot)$,
we say the protocol $\Pi$ is a secure QPFE.

\subsection{Comparison with the classical analogue}

To understand the similarity and difference between QPFE and PFE, we compare their definitions in this section.
Recall that PFE has the following ideal functionality \cite{liu2022making}.  

Pre-agreement: $\#$ gates, $\#$ output wires, $\#$ input wires.

Input:  Party A: a circuit $C_f$ encoding the function $f$, its input $x_A$; Party B: its input $x_B$.

\begin{enumerate}
\item If either party sends abort to the ideal functionality $\mathcal{F}$, $\mathcal{F}$ broadcasts $\bot$ and halts.
\item If the $C_f$ received from $A$ satisfies the pre-agreement, the ideal functionality stores $C_f$.
\item Check if $x_A$ and $x_B$ are received and $C_f$ is stored. If not, the process halts.
The adversary $\mathcal{A}$ sends a message continue or abort to $\mathcal{F}$. 
 If $\mathcal{F}$ receives continue, it computes $C_f(x_A, x_B)$ and delivers $C_f(x_A, x_B)$ to all parties,
 otherwise the process halts.
\end{enumerate}

By comparing the ideal functionalities of PFE and QPFE, it is clear that these two tasks are
almost the same, except that the inputs and function/algorithm in QPFE are quantum, 
while the inputs and function/algorithm in PFE are classical.
Note however that the two tasks of QPFE and PFE are not directly comparable
in terms of their performance as they 
consider different objects (classical objects for PFE and quantum objects for QPFE).
While there might be quantum algorithms which are more advantageous than any classical algorithm, 
this is orthogonal to our work and outside the scope of our paper.

\subsection{Variants}

 There are several variants of QPFE definitions that can be considered. 
Before explaining these variants, note however that these discussions of variants are not used in the rest of the paper.
The reader may skip this section and directly go to the next section with no loss.
 
For the first variant, the private function might be distributed among several parties instead of just one.
For this case, the ideal functionality should be modified as follows: 
it should first collect part of the private function $\mathcal{E}^i$ from party $i$ for all $1\le i\le n$. Then
the ideal functionality use a combining mechanism $\textsf{Combine}(\cdot)$ to yield the complete quantum function
\begin{equation}
\mathcal{E}= \textsf{Combine}(\mathcal{E}^1, \cdots, \mathcal{E}^n).
\end{equation}
The rest procedures of the ideal functionality is the same as the standard ideal functionality of QPFE.

For the second variant, some participants may go offline during the process. A good QPFE protocol should ideally address the security issue in this situation.
For this case, the ideal functionality should be modified as follows:
Instead of sending Party $i$ the result $\mathcal{E}^{(i)}(\rho_1,\cdots,\rho_n)$ directly,
the ideal functionality first waits for a message from  Party $i$ which can either be ``online''
or ``offline''. Only if the message ``online'' is received, the ideal functionality 
sends $\mathcal{E}^{(i)}(\rho_1,\cdots,\rho_n)$ to Party $i$.

For the third variant, both the function and the inputs in QPFE might be time-dependent instead of invariant with respect to time.
For this case, the ideal functionality should be modified as follows:
First of all, the inputs from all parties are sent to the ideal functionality simultaneously 
and instantly (or within a short time interval $\Delta_1$). 
Second, the ideal functionality performs the calculation instantly (or within a short time interval $\Delta_2$).
Third, the ideal functionality sends the outputs to the parties simultaneously and instantly (or within a 
short time interval $\Delta_3$). The most important difference from the standard definition is that
 every step is timestamped in this scenario.

\section{Scheme of QPFE}
\label{sec:protocol}

In this section, we present two schemes of QPFE. In Sec.~\ref{sec:firstprotocol}, we present 
the first scheme of QPFE, which utilizes MPQC as a black box. In Sec.~\ref{sec:secondprotocol},
we present the second scheme of QPFE, which builds directly on smaller primitives. 
In Sec.~\ref{sec:reuse}, we extend the schemes to the reusable scenario, where multiple runs 
of QPFE are executed. We show that the average communication cost can be lowered 
in this multiple-run scenario.

\subsection{First protocol}
\label{sec:firstprotocol}

In this section, we present the first QPFE scheme. In Sec.~\ref{sec:1protodesc}, we describe the procedures 
of the first QPFE scheme.  In Sec.~\ref{sec:1protoimpl}, we give the implementation
of the first QPFE scheme on a quantum computer.  In Sec.~\ref{sec:1protoeg}, we give an example 
of the first QPFE scheme for specific quantum functions and inputs. 

\subsubsection{Protocol description}
\label{sec:1protodesc}
The critical idea is to convert the quantum function $\mathcal{E}$ into a quantum state $\rho_\mathcal{E}$. Then 
we can apply MPQC to obtain a secure protocol of QPFE. This can indeed be done by the Choi-Jamiolkowski (C-J) transform \cite{jamiolkowski1972linear,choi1975completely}.

Let us review the C-J transform. To describe the C-J transform, we need to first define the Choi matrix $\Upsilon_\mathcal{E}$ corresponding to a quantum channel $\mathcal{E}$, which has the form
\begin{equation}
\label{eq:Choi}
\Upsilon_\mathcal{E} = ( I \otimes \mathcal{E}  )  (\ket{\psi} \bra{\psi} ) =  \sum\limits_{k,l =1 }^d  \ket{k}\bra{l} \otimes \mathcal{E} (\ket{k}\bra{l}),
\end{equation}
where $\ket{\psi}=\sum_k \ket{kk} $ is an unnormalized Bell state and $\Upsilon_\mathcal{E}$ is positive definite as $\mathcal{E}$ is a completely positive map.
Note that $\Upsilon_\mathcal{E}$ is not necessarily a quantum state as its trace may not be 1. 
We are now ready to present the C-J transform. It can be expressed as
\begin{equation}
 \mathcal{E}(\rho) = \textrm{Tr}_{A} [ ( \rho^T_A  \otimes \mathbbm{1}_B  ) \Upsilon_\mathcal{E}  ],
\end{equation}
where $\rho^T_A \in \mathcal{H}_A$,  $ \mathbbm{1}_B \in \mathcal{H}_B$,  $\Upsilon_\mathcal{E} \in  \mathcal{H}_A \otimes   \mathcal{H}_B$, $ \mathcal{H}_A$ and $ \mathcal{H}_B$ are the Hilbert space over the systems $A$ and $B$ respectively, and $ \textrm{Tr}_{A}$ is the partial trace over the $A$ system.
The C-J transform establishes an equivalence between the channel $\mathcal{E}$ and its corresponding Choi matrix $\Upsilon_\mathcal{E}$.

Based on the C-J transform, we can define the following $G$ functionality 
\begin{equation}
 G(\rho_1,\rho_2, \rho_3) =  \textrm{Tr}_{A,B} [ ((  \rho_1 \otimes \rho_3 ) \otimes \mathbbm{1}_C) \rho_2 ],
\end{equation}
where $\rho_1 \in \mathcal{H}_A$,  $\rho_3 \in \mathcal{H}_B$, $\mathbbm{1}_C \in  \mathcal{H}_C$,  and $\rho_2 \in \mathcal{H}_{ABC}$.
To utilizes the $G$ functionality, party 1 prepares his inputs as $(\rho_1, \rho_2) = (\rho_A^T, \Upsilon_\mathcal{E}/ tr(\Upsilon_\mathcal{E}))$ and party 2 prepares his input as $\rho_3 = \rho_B ^ T$.  The two parties then perform an MPQC with the $G$ functionality on their private inputs. According to the property of MPQC, the function $\mathcal{E}$ of party 1 is kept secret, 
as well as the private quantum states of the two parties.

Note that since the transpose operation ($\rho \to \rho^T$) is not completely positive, it cannot be realized as a quantum map. Hence, we assume 
the classical description of $\rho$ is known.  Compared with the protocol that utilizes classical PFE solely, the cost of the online 
phase of this protocol becomes polynomial instead of exponential. This is because the preparation of $\rho_A^T$ and $\rho_B^T$ involves no communication and
can all be done offline by the parties themselves.

For ease of notation, let us define $\rho = \rho_A \otimes \rho_B$. For proving the correctness of the protocol, let us show that $G(\rho_1,\rho_2, \rho_3) =  \mathcal{E}(\rho)/\textrm{Tr}(\Upsilon_\mathcal{E})$ as follows,
\begin{equation}
\begin{aligned}
G(\rho_1,\rho_2, \rho_3)  & = \textrm{Tr}_{A,B} [((  \rho_1 \otimes \rho_3 ) \otimes \mathbbm{1}) \rho_2] \\
& = \textrm{Tr}_{A,B}[ ((  \rho_A^T \otimes \rho_B^T ) \otimes \mathbbm{1}) \Upsilon_\mathcal{E}/\textrm{Tr}(\Upsilon_\mathcal{E}) ] \\
& = \textrm{Tr}_{A,B} [ ((  \rho_A \otimes \rho_B )^T \otimes \mathbbm{1}) \Upsilon_\mathcal{E}/\textrm{Tr}(\Upsilon_\mathcal{E}) ] \\
& = \textrm{Tr}_{A,B} [(\rho^T \otimes \mathbbm{1}) \Upsilon_\mathcal{E} /\textrm{Tr}(\Upsilon_\mathcal{E})] \\
& =  \mathcal{E}(\rho)/ \textrm{Tr}(\Upsilon_\mathcal{E}).
\end{aligned}
\end{equation}
In the second equality, we have utilized the identity $(A \otimes B) ^T = A^T \otimes B^T $. 

This protocol can be easily extended to the multi-party case. 
We only need to change the $G$ functionality to
\begin{equation}
\begin{aligned}
 G(\rho_1,\rho_2, \cdots, \rho_{n+1})  = & \textrm{Tr}_{A,B} [ ((  \rho_1 \otimes \rho_3\otimes \rho_4 \otimes  \\
 & \cdots \otimes  \rho_{n+1} ) \otimes \mathbbm{1}) \rho_2 ].
 \end{aligned}
\end{equation}
Here, $\rho_i$ ($3\le i \le n+1$) belongs to party $i$. $\rho_1$ and $\rho_2$ belong to party 1.

\subsubsection{Implementation on a quantum computer}
\label{sec:1protoimpl}

Now, we consider how to implement the functionality $G$.
First we make a decomposition on the state $\rho = \rho_1 \otimes \rho_3$.
It can be decomposed as 
\begin{equation}
\rho = \sum p_i \ket{\psi_i} \bra{\psi_i},
\end{equation}
which can be viewed as a part of a positive operator-valued measure (POVM). 
The reason is as follows.
For each  $i$, we can expand $ \ket{\psi_i}$ into a basis of $\mathcal{H}_{AB}$ as $\{ \ket{\psi_i},  \ket{\psi_i^2},  \cdots \ket{\psi_i^d} \}$, 
where $d$ is the dimension of $\mathcal{H}_{AB}$. By the definition of a basis, we have 
\begin{equation}
 \ket{\psi_i}\bra{\psi_i} +\sum_{j=2}^d \ket{\psi_i^j} \bra{\psi_i^j}= I.
\end{equation}
Hence, 
\begin{equation}
I-\rho = \sum_i p_i \sum_{j=2}^d \ket{\psi_i^j} \bra{\psi_i^j}.
\end{equation}
This expression is positive definite as each summand is positive definite. Hence, $\{ \rho , I-\rho\}$ forms a POVM.

By preparing $\rho_2$ and  measuring the subsystem $\mathcal{H}_{AB}$ with the POVM  $\{   (I_{AB}-\rho ), \rho   \} $,
which corresponds to the outcomes 0 and 1 respectively. 
We post-select the outcome of the remaining quantum system when we get the measurement result $\rho \otimes I_C$.
Then we  obtain the quantum state $\mathcal{E}(\rho) $.
The procedure is illustrated in Fig.~\ref{fig:schematics}.
\begin{figure}[htb]
\centering 
\includegraphics[width=8cm]{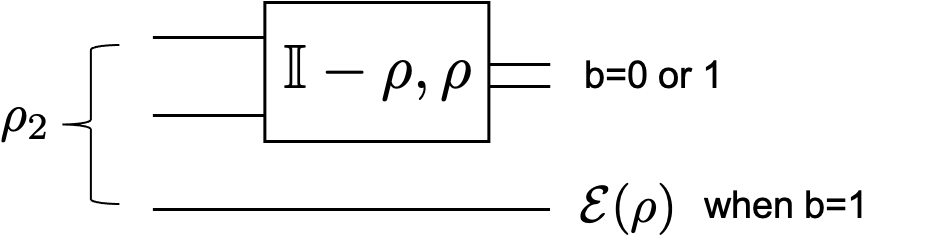}
\caption{An illustration of the implementation of the first QPFE scheme.}
\label{fig:schematics}
\end{figure}
Note that this procedure is probabilistic. The success probability is given by  $1 / \textrm{Tr}(\Upsilon_\mathcal{E}) $.

\subsubsection{Example}
\label{sec:1protoeg}

Let us now give an example of the first QPFE scheme.
In the example, we consider the depolarizing channel, which has the form 
\begin{equation}
\label{eq:depolarizing}
\mathcal{E}(\rho) = p \rho + (1-p) \frac{\mathbbm{1}}{2}.
\end{equation}
Here party $A$ does not have any input, party $B$ has input $\rho_B = \rho$. 
We decompose $\rho$ as 
\begin{equation}
\rho = \sum p_i \ket{\psi_i} \bra{\psi_i}.
\end{equation}
Apparently, $(\rho, \mathbbm{1}- \rho )$ form a complete set of POVM when $\rho$ is a qubit.
This is because we can  write $ \mathbbm{1} - \rho$ as 
\begin{equation}
\mathbbm{1} -  \rho = \sum p_i \ket{\psi_i^\bot} \bra{\psi_i^\bot}.
\end{equation}
This matrix is a mixture of pure states. In addition, it has trace 1 as all $\ket{\psi_i^\bot} \bra{\psi_i^\bot}$ have trace 1 and 
$\sum p_i =1$.
This phenomenon can easily be extended to $n$ qubits.
The difference is that, while 1 qubit corresponds to 2 bases, $n$ qubits correspond to $2^n$ bases. 
For concreteness, let us take 
\begin{equation}
\label{eq:rhoeg}
\rho = \ket{+}\bra{+}.
\end{equation}
We have the following claim:
\begin{theorem}
For $\rho$ defined in Eq.~\eqref{eq:rhoeg} and $\mathcal{E}$ defined in Eq.~\eqref{eq:depolarizing}, the identity 
\begin{equation}
\textrm{Tr}(\Upsilon_\mathcal{E}) G( \rho^T, \emptyset, \Upsilon_\mathcal{E}/\textrm{Tr}(\Upsilon_\mathcal{E}))  = \mathcal{E}(\rho)
\label{eq:example}
\end{equation}
holds.
\end{theorem}

\begin{proof}
First, the right-hand side (RHS) of Eq.~\eqref{eq:example} equals 
\begin{equation}
p \ket{+}\bra{+} + (1-p) \frac{\mathbb{I}}{2}.
\end{equation}
We are left to show that the left-hand side (LHS) also equals this value.

Before calculating the LHS, we first compute $\Upsilon_\mathcal{E}$ explicitly.
By Eq.~\eqref{eq:depolarizing} and Eq.~\eqref{eq:Choi},
we have 
\begin{equation}
\Upsilon_\mathcal{E} = 2 p \ket{\Phi^+}\bra{\Phi^+} + 2 (1-p) ( \frac{\mathbb{I}}{4}).
\end{equation}
More precisely, the depolarizing channel $\mathcal{E}$ satisfies
\begin{equation}
\mathcal{E}(I)=I, \mathcal{E}(\sigma_x)=p \sigma_x, \mathcal{E}(\sigma_y) = p \sigma_y,  \mathcal{E}(\sigma_z) = p \sigma_z.
\end{equation}
Hence
\begin{equation}
\begin{aligned}
\mathcal{E}(\ket{0}\bra{0}) = (1-p) I/2 + p/2 \sigma_z,  \quad  &  \mathcal{E}(\ket{0}\bra{1}) = p \ket{0}\bra{1} ,  \\
  \mathcal{E}(\ket{1}\bra{1}) = (1-p) I/2 - p/2 \sigma_z, \quad    &   \mathcal{E}(\ket{1}\bra{0}) = p \ket{1}\bra{0}.
\end{aligned}
\end{equation}
Substituting these into 
\begin{equation}
\Upsilon_\mathcal{E}  = \sum\limits_{k,l} \mathcal{E}(\ket{k}\bra{l}) \ket{k}\bra{l},
\end{equation}
we get 
\begin{equation}
\Upsilon_\mathcal{E} = 2 p \ket{\Phi^+}\bra{\Phi^+} + 2 (1-p) ( \frac{\mathbb{I}}{4}).
\end{equation}

With the expression of $\Upsilon_\mathcal{E}$ at hand, we are now ready to calculate the LHS of  Eq.~\eqref{eq:example}. It equals 
\begin{equation}
\begin{aligned}
  & \textrm{Tr}_1 ( (\rho^T \otimes \mathbb{I})  \Upsilon_\mathcal{E})  \\
=&\textrm{Tr}_1 ( (\rho^T \otimes \mathbb{I})  2p\ket{\Phi^+}\bra{\Phi^+}) +  \textrm{Tr}_1 ( (\rho^T \otimes \mathbb{I}) 2 (1-p) ( \frac{\mathbb{I}}{4}))   \\
  =& 2p \textrm{Tr}_1 ( (\rho^T \otimes \mathbb{I})  \ket{\Phi^+}\bra{\Phi^+})  + 2 (1-p)  \textrm{Tr}_1 ( (\rho^T \otimes \mathbb{I})  ( \frac{\mathbb{I}}{4})).  \\
\end{aligned}
\end{equation}
The first term amounts to measure $\ket{+}$ on the first qubit of the Bell state $ \ket{\Phi^+}$ and then to take the second qubit.
Apparently, the second qubit should be $\ket{+}$ and the probability that the first qubit has measurement outcome $\ket{+}$ is $1/2$, 
so the first term becomes $ p \ket{+}\bra{+}$. For the second term, it amounts to measure $ \ket{+}$ on the first qubit of $I/4$ 
and then take the second qubit. The second qubit should be $I/2$ and the probability that the first qubit has measurement outcome $\ket{+}$ is $1/2$ too,
so the first term becomes $ (1-p) I/2$. Overall, the LHS of Eq.~\eqref{eq:example} becomes 
\begin{equation}
p \ket{+}\bra{+} + (1-p) \frac{\mathbb{I}}{2},
\end{equation}
which equals the RHS of Eq.~\eqref{eq:example}.
\end{proof}

%

\subsection{Second protocol}
\label{sec:secondprotocol}

Without loss of generality, we restrict the private circuit of Bob to be a Clifford + Measurement (C+M) circuit  $Q$ acting 
on two parties' private inputs and some ancillary zero states, $\ket{0}$, and T states, $(\ket{0}+e^{i\pi/4} \ket{1})/\sqrt{2}$.

We first describe the classical functionality used in the second protocol.
Before the classical functionality, Alice randomly selects two Cliffords $C_{A,in}$ and $C_{A,out}$ and  Bob randomly selects a Clifford $C_{B,in}$.
Let $Q_B=Q[C_{A,out}]$ be the circuit that performs $C_{A,out}$ on Alice's output after applying the circuit $Q$.
Let $(E_0, \widetilde{Q_B})$ be the garbled circuit of $Q_B$.
Bob and Alice calculate the garbled circuit through a classical 2PC with syntax ($\textsf{2PC}_1, \textsf{2PC}_2$), as detailed in Protocol \ref{Fig:Procedure}.
Here $C_{A,in}, C_{A,out}$ are only known to Alice, and $Q, C_{B,in}$ are only known to Bob. 

\begin{algorithm}
\caption{\textsc{Classical functionality}}
\begin{flushleft}
\textit{Alice Input}: $C_{A,in}, C_{A,out}$.  

\textit{Bob Input}: $C_{B,in}$, $Q$. 
\end{flushleft}
\begin{algorithmic}[1]
\STATE
 Sample $(E_0, \widetilde{Q_B}) = \textsf{QGarble}(Q[C_{A,out}])$
\STATE
Let $W= E_0 (I \otimes C_{B,in}^{-1} \otimes I) C_{A,in}^{-1}$. 
\end{algorithmic}
\begin{flushleft}
\textit{Bob Output}: (1) $W$; (2) $\widetilde{Q_B}$.
\end{flushleft}
\label{Fig:Procedure}
\end{algorithm}

We are now ready to present the second protocol. 
Let the quantum input of Alice be $x_A$ and that of Bob be $x_B$.
Denote $k$ zero states as $0^k$ and $k$ T states as $T^k$.
The complete protocol has three rounds which are as follows:
\begin{enumerate}
\item Bob sends Alice the messages $\mathbf{m_{B,1}}= C_{B,in}(x_B,0^\lambda)$ and $m_{B,1}=\textsf{2PC}_1(C_{B,in},Q)$, where $\lambda$ is the security parameter and $C_{B,in}$ is a unitary that acts on the joint quantum state of $x_B$ and $0^\lambda$.
\item Alice sends Bob the messages $\mathbf{m_A} = C_{A,in}( x_A, \mathbf{m_{B,1}}, T^k, 0^{k} )$ and $m_A=\textsf{2PC}_2(m_{B,1}, C_{A,in},C_{A,out})$, where $k$ is the number of ancilliary $T$ states.
\item Bob first gets $(W, \widetilde{Q_B})=\textsf{2PC}_{out}(m_A)$ with the first two rounds. 
Bob then calculates 
\begin{equation}
 \begin{aligned}
 W(\mathbf{m_A} ) &= E_0 (\mathbb{I}\otimes C^{-1}_{B,in} \otimes \mathbb{I}) C^{-1}_{A,in} (\mathbf{m_A}) \\
 &= E_0(x_A, x_B, T^k, 0^{k+\lambda}).
 \end{aligned}
\end{equation}
Together with $\widetilde{Q_B}$, Bob gets $(C_{A,out}(y_A, 0^\lambda),y_B)$.
Let $\mathbf{m_{B,2}}= C_{A,out}(y_A, 0^\lambda)$. Bob sends $\mathbf{m_{B,2}}$ to Alice. Alice then recovers $y_A$ from $\mathbf{m_{B,2}}$. 
\end{enumerate}

Note an important change compared to 2PQC is that in the classical functionality of QPFE, $Q$ is Bob's private information, while in that of 2PQC, $Q$ is common information 
of both parties.

\subsection{Reuseable extensions}
\label{sec:reuse}

In this section, we consider the scenario where multiple runs of QPFE are executed. 
In this case, Bob's first message can be reused.
This reduces subsequent executions to two rounds and reduces the average communication and computational consumption.
In the initial executions, either Bob sends multiple quantum messages $\mathbf{m}$ for use in later executions, or sends only classical messages, which
can be easily copied for use in  later executions.

\section{Security proofs of QPFE}
\label{sec:securityproof}

\subsection{Security proof of the first scheme}

Since party 1 can change his private function $\mathcal{E}$ to $\mathcal{E}'(\rho_A, \rho_B) =\rho_B$ 
which completely undermines party 2's security, we only consider security against malicious party 2 and security against semi-honest party 1.

Let $I_{2PQC}[G, \rho_A, \mathcal{E}, \rho_B]$ denote the ideal functionality  of 2PQC for the function $G(\rho_1,\rho_2, \rho_3)$, where the inputs are $(\rho_1, \rho_2, \rho_3) 
= (\rho_A^T, \Upsilon_\mathcal{E}/ tr(\Upsilon_\mathcal{E}), \rho_B ^ T)$. In the following proof, we will write $I[\rho_A, \mathcal{E}](\cdot) = I_{2PQC}(G, \rho_A, \mathcal{E}, \cdot)$ 
and $I[\rho_B](\cdot, \cdot) = I_{2PQC}(G, \cdot, \cdot, \rho_B)$.

\begin{proof}[Security Proof]

First, we give the simulator that interacts with a malicious party 2.
The simulator gives $\rho_B$ to the ideal functionality $I[\rho_A, \mathcal{E}](\cdot)$ and receives $\tau_B$. 
The simulator then sends  $\tau_B$ to party 2.
Since the simulator is ignorant of party 1's private inputs, the simulator is secure. 
We denote the real execution as $\mathcal{H}_0$ and the simulated execution as $\mathcal{H}_1$.

$\mathcal{H}_0 \approx_c \mathcal{H}_1$: This results from the security of 2PQC.

Next, we give the simulator that interacts with a semi-honest party 1.
The simulator gives $\rho_A,\mathcal{E}$ to the ideal functionality $I[\rho_B](\cdot, \cdot)$ and receives $\mathcal{E}(\rho_A, \rho_B)$.
Since party 1 is semi-honest, $\mathcal{E}(\rho_A, \rho_B)$ is indeed the correct output. The adversary has no extra information
in this simulated world other than the circuit output and its inputs.
We are left to show the real world is indistinguishable from the ideal world. This again results from the security of 2PQC. 
\end{proof}

\subsection{Security proof of the second scheme}

\textbf{Security against malicious Alice:} 
Let the simulator of the classical 2PC be $(Sim_A^{(1)}, Sim_A^{(2)})$.
Consider the following simulator Sim for the malicious Alice:
\begin{enumerate}
\item The simulator samples $C_{B,in}$ randomly and computes $\mathbf{m}_{B,1} = C_{B,in}(0^{n_B}, 0^\lambda)$ and $m_{B,1}=Sim_A^{(1)}(1^\lambda)$. The simulator sends Alice $(\mathbf{m}_{B,1},m_{B,1})$. 
\item Alice sends Sim $( \mathbf{m_A},m_A)$.  Sim then calculates $ (C_{A,in}, C_{A,out},E_0)  \leftarrow 2PC.Sim_A^{(2)} ( m_A)$. 
\item By $(C_{A,in}, C_{A,out},C_{B,in})$, the simulator computes
\begin{equation}
 ( x_A', x_B') := U_{dec} (\mathbf{m_A}),
\end{equation}
where $ U_{dec} $ is $(\mathbb{I}\otimes C^{-1}_{B,in} \otimes \mathbb{I}) C^{-1}_{A,in}$.
\item Sim sends $x_A'$ to $I[x_B, Q]( \cdot)$ (Here $Q$ is the description of the circuit known only to Bob), gets $y_A$ back, computes $\widehat{y_A}= C_{A,out}(y_A, 0^\lambda)$, 
sends $\widehat{y_A}$ to Alice and outputs whatever Alice outputs.
\end{enumerate}

The simulator is secure as it is completely ignorant of Bob's private inputs.
We now define a series of hybrid executions:
\begin{itemize}
\item  $\mathcal{H}_0$: The real execution.
\item   $\mathcal{H}_1$: This hybrid is the same as the real execution, except that $m_{B,1}$ is obtained through $2PC.Sim_A^{(1)}$ and Bob gets $(C_{A,in},C_{A,out},E_0)$ through $2PC.Sim_A^{(2)}(m_2)$. Bob performs the garble procedure using this information and continues as in the real execution. 
\item   $\mathcal{H}_2$: This hybrid is the same as $\mathcal{H}_1$ except that, instead of the garbling procedure, Bob directly evaluates $Q[C_{A,out}]$ on the input $(x_A, x_B)$.
\item  $\mathcal{H}_3$:  This hybrid is the same as $\mathcal{H}_2$ except that  $\mathbf{m_1}= C_{B,in}(x_B,0^\lambda)$ is replaced by $\mathbf{m_1}= C_{B,in}(0^{n_B},0^\lambda)$.
\item  $\mathcal{H}_4$: This hybrid is the same as $\mathcal{H}_3$ except that instead of evaluating the circuit directly, send $x_A'$ to $I[x_B, Q]( \cdot)$  to receive $y_A$ and then compute $\widehat{y_A}= C_{A,out}(y_A, 0^\lambda)$. This is exactly the simulated execution.
\end{itemize}

We now show the equivalences between these hybrids: 
\begin{itemize}
\item $\mathcal{H}_0 \approx_c \mathcal{H}_1$: This results from the security of 2PC.
\item $\mathcal{H}_1 \approx_s \mathcal{H}_2$: This results from the correctness of the quantum garbled circuit.
\item $\mathcal{H}_2 \approx_s \mathcal{H}_3$: This results from the security of the Clifford authentication code.
\item  $\mathcal{H}_3 \approx_s \mathcal{H}_4$:  
We have that computing $Q[C_{A,out}]$ on $(x_A, x_B,  T^k, 0^{k+\lambda})$ is 
statistically indifferent from the ideal functionality, namely computing first $(y_A, y_B) = Q(x_A, x_B,  T^k, 0^{k+\lambda})$ and then $\widehat{y_A}= C_{A,out}(y_A, 0^\lambda)$. Hence the two hybrids are equivalent.
\end{itemize}

\textbf{Security against malicious Bob:}  Next we prove security against malicious Bob.
Note that in this case, we require that Bob has no output, otherwise Bob can simply change the function such that $Q'(x_A,x_B)=x_A$ to obtain extra information on Alice's private inputs 
in addition to $Q(x_A, x_B)$. Note that Bob cannot copy Alice's output before sending Alice's output to Alice because of the no-cloning theorem.
Denote the simulator for classical 2PC as $(Sim_B^{(1)}, Sim_B^{(2)})$.

Consider the following simulator Sim for Bob.

\begin{enumerate}[leftmargin=*]
 \item  First round
       \begin{itemize}[leftmargin=*]
         \item[-] Get $(\mathbf{m}_{B,1},m_{B,1})$ from Bob. Compute $(C_B,Q) = 2PC.Sim_B^{(1)}(m_{B,1})$.
     \end{itemize}
 \item  Second round
      \begin{itemize}[leftmargin=*]
         \item[-] Sample a random Clifford  $C_{A,out}\leftarrow \mathcal{C}$ and compute $\widetilde{y_A}=C_{A,out}(0^{m_A})$.
         \item[-] Compute $(m_{inp}, \widetilde{Q_B}) \leftarrow \textsf{QGSim}(\widetilde{y_A}, \{n_i, k_i\}, 1^\lambda)$.
         \item[-] Sample a random Clifford $U_{dce} \leftarrow \mathcal{C}$. Here, $U_{dce}$ corresponds to $W$ in the protocol.
         \item[-] Compute $\mathbf{m}_{A,2} = U_{dce}^\dagger (m_{inp}, 0^k, T^k)$.
         \item[-] Compute $m_{A,2} = 2PC.Sim^{(2)}_B (U_{dce}, \widetilde{Q_B}) $.
         \item[-] Send $(\mathbf{m}_{A,2}, m_{A,2})$ to Bob.
     \end{itemize}
 \item  Third round
      \begin{itemize}[leftmargin=*]
         \item[-] Get $\widehat{y_A}$ from Bob and check $C_{A,out}^{\dagger}(\widehat{y_A})$.
         \item[-] If the last bits of $C_{A,out}^{\dagger}(\widehat{y_A})$ are not all zero, abort. Otherwise send ok to the ideal functionality.
         \item[-] Output $\widehat{y_A}$.
     \end{itemize}
\end{enumerate}

The simulator does not know Alice's real input $x_A$ and real output $y_A$.
Next we show Bob's view in this ideal world is indistinguishable from his view in the real world. 

We define a series of hybrids as follows:
\begin{itemize}
\item $\mathcal{H}_0$:  The real world.

\item $\mathcal{H}_1$: This hybrid is almost identical to $\mathcal{H}_0$ except how $m_{A,2}$ is calculated. In $\mathcal{H}_0$, 
$m_{A,2}$ is calculated from $m_{B,1}$ through $\textsf{2PC}_2$. In this hybrid, $C_B$ is first recovered from $m_{B,1}$ by $2PC.Sim^{(1)}_B$,
which is fed into $2PC.Sim^{(2)}_B$ to obtain $m_{A,2}$.

\item $\mathcal{H}_2$: This hybrid is almost identical to $\mathcal{H}_1$ except how $\mathbf{m}_{A,2}$ is calculated. In $\mathcal{H}_1$, according to the
definition of $U_{dec}$, we have
\begin{equation}
U_{dec}(\mathbf{m}_{A,2}) = (E_0(\mathbf{x}_A, \mathbf{x}_B',  T^k, 0^{k+\lambda}), 0^k,  T^k).
\end{equation}
In this hybrid, we first obtain $\mathbf{x}_B' = C_B^\dagger (\mathbf{m}_{B,1})$, then sample $U_{dec} \in \mathcal{C}$, and finally compute
\begin{equation}
\mathbf{m}_{A,2} = U_{dec}^\dagger(E_0(\mathbf{x}_A, \mathbf{x}_B',  T^k, 0^{k+\lambda}), 0^k,  T^k).
\end{equation}

\item $\mathcal{H}_3$: This hybrid is almost identical to $\mathcal{H}_2$ except calculating the output directly instead of using the quantum garbled circuit.
More concretely, since we already know $\mathbf{x}_B'$, we can calculate 
\begin{equation}
 \widehat{y_A} \leftarrow Q[C_{A,out}](x_A, x_B', T^k, 0^{k+\lambda}).
\end{equation}
Then we obtain 
\begin{equation}
(\mathbf{m}_{inp},  \widetilde{Q_B}) \leftarrow \textsf{QGSim}( \widehat{y_A} , \{n_i, k_i\}, 1^\lambda).
\end{equation}
We then substitute $E_0(\mathbf{x}_A, \mathbf{x}_B', T^k, 0^{k+\lambda})$ by $\mathbf{m}_{inp}$, namely, we calculate $\mathbf{m}_{A,2}$
as 
\begin{equation}
\mathbf{m}_{A,2} = U_{dec}^\dagger(\mathbf{m}_{inp}, T^k, 0^{k+\lambda}).
\end{equation}

\item $\mathcal{H}_4$: This hybrid is almost identical to $\mathcal{H}_3$ except the following change. Recall that in $\mathcal{H}_3$,
$Q[C_{A,out}](x_A, x_B', T^k, 0^{k+\lambda})$ is calculated through a two-step process:
\begin{equation}
(x_A, x_B', T^k, 0^{k+\lambda}) \to (y_A, 0^\lambda, C_{A,out}) \to C_{A,out}(y_A, 0^\lambda).
\end{equation}
In this hybrid, we replace the result of the second step by $C_{A,out}(0^{m_A+\lambda})$,
where $m_A$ is the number of qubits of $y_A$.
This is exactly the ideal world.
\end{itemize}

Next we show  the equivalence between these hybrids:
\begin{itemize}
\item $\mathcal{H}_0 \approx_c \mathcal{H}_1$: This results from the security of 2PC.
\item $\mathcal{H}_1 \approx_s \mathcal{H}_2$: This stems  from the definition of  $U_{dec}$.
\item $\mathcal{H}_2 \approx_c \mathcal{H}_3$: This results from the security of the quantum garbled circuit.
\item $\mathcal{H}_3 \approx_s \mathcal{H}_4$:   This results from the security of the Clifford authentication code.
\end{itemize}

\section{Experimental demonstration}
\label{sec:exp}

In this section, we perform an experimental demonstration of the QPFE scheme based on the IBM quantum experience platform \cite{ibmQ}.
We consider a two-party quantum functionality in Fig.~\ref{fig:1} and call the two parties $P_1$ and $P_2$. Here, the first qubit belongs to $P_1$ and the second 
qubit belongs to $P_2$. At the end, $P_2$ gets a classical bit as the output, while $P_1$ gets no output.
For example, when $P_1$'s initial qubit is $\ket{0}$ and $P_2$'s initial qubit is $\ket{+}$, then by direct calculation one can obtain that 
$P_2$ will get the output $\frac{1}{2} \{ 0\} , \frac{1}{2} \{ 1\}$ and $P_1$ gets no output. Here,  $\frac{1}{2} \{ 0\} , \frac{1}{2} \{ 1\}$ means 
that with a probability $1/2$, the value is 0 and  with a probability $1/2$, the value is 1.

\begin{figure}[htb]
\centering \includegraphics[width=8cm]{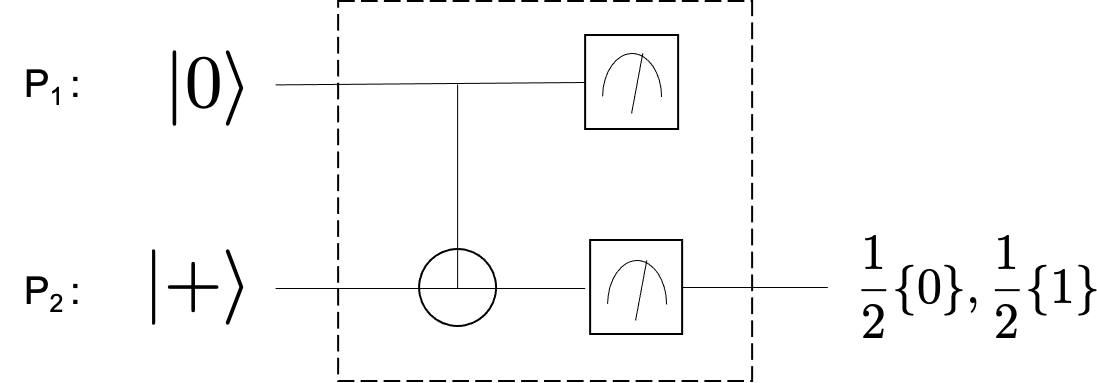}
\caption{Example private functionality and inputs.}
\label{fig:1}
\end{figure}

There are three properties that a correct QPFE protocol should satisfy:
\begin{itemize}
\item $P_2$ is not certain whether $P_1$ is $\ket{0}$ or $\ket{1}$.
\item $P_1$ is not certain $P_2$ is $\ket{+}$ or $\ket{-}$.
\item The third party $P_3$ is ignorant of the quantum states of $P_1$ and $P_2$.
\end{itemize}
For example, a wrong protocol that does not satisfy these three properties can be that $P_2$ directly outputs $\frac{1}{2}  \{ 0\}$,  $\frac{1}{2}  \{ 1\}$. 
In that case, the third party $P_3$ would know that the quantum state of $P_2$ is
$\frac{1}{\sqrt{2}} (\ket{0}+e^{i\phi}\ket{1})$ or their combinations. This is 
in conflict with the third property.

Now we turn to the design of the correct protocol.
For ease of implementation, we set the security parameter to be $\lambda=0$.
 The protocol design consists of four steps.
In the first step, we add a randomized $Z^s$ gate as a substitute of $P_2$'s measurement, as illustrated in Fig.~\ref{fig:expscheme}(a).
Here, $s$ is randomly chosen from $\{0,1\}$.
Call the resulting circuit $F_1$.
In the second step, we apply a randomized Pauli rotations $X^{a_1}Z^{b_1}$ and $X^{a_2}Z^{b_2}$  to both $P_1$ and $P_2$'s private inputs,
and then apply $F_1$ to these two qubits and obtain the outcome $\rho_3$, as illustrated in Fig.~\ref{fig:expscheme}(b). 
Here $P_2$ is ignorant of $a_1$ and $b_1$, and $P_1$ is ignorant of $a_2$ and $b_2$.
By the QPFE scheme, we can apply an $X^{a_3}$ operation on $\rho_3$ and then perform computational-basis measurements
to get the correct circuit output $F(\rho_1, \rho_2)$, as illustrated in Fig.~\ref{fig:expscheme}(c). This is the fourth step. 
In the third step, we use classical PFE to calculate $a_3$, as illustrated in Fig.~\ref{fig:expscheme}(d). 
Here, we are not using classical MPC as that will leak sensitive information on Bob's private circuit. 
The classical PFE takes as inputs $a_1,b_1,a_2,b_2,s$ and outputs $a_3$.
The combined protocol is summarized in Fig.~\ref{fig:expscheme}(e). 
Here, $\hat{\rho_1}$ and $\hat{\rho_2}$ are locally processed by $P_1$ and $P_2$ respectively.
Importantly, to keep Bob's private circuit secret, the encrypted circuit $\hat{F}$ is held by Bob.

\begin{figure}[htb]
\centering \includegraphics[width=9 cm]{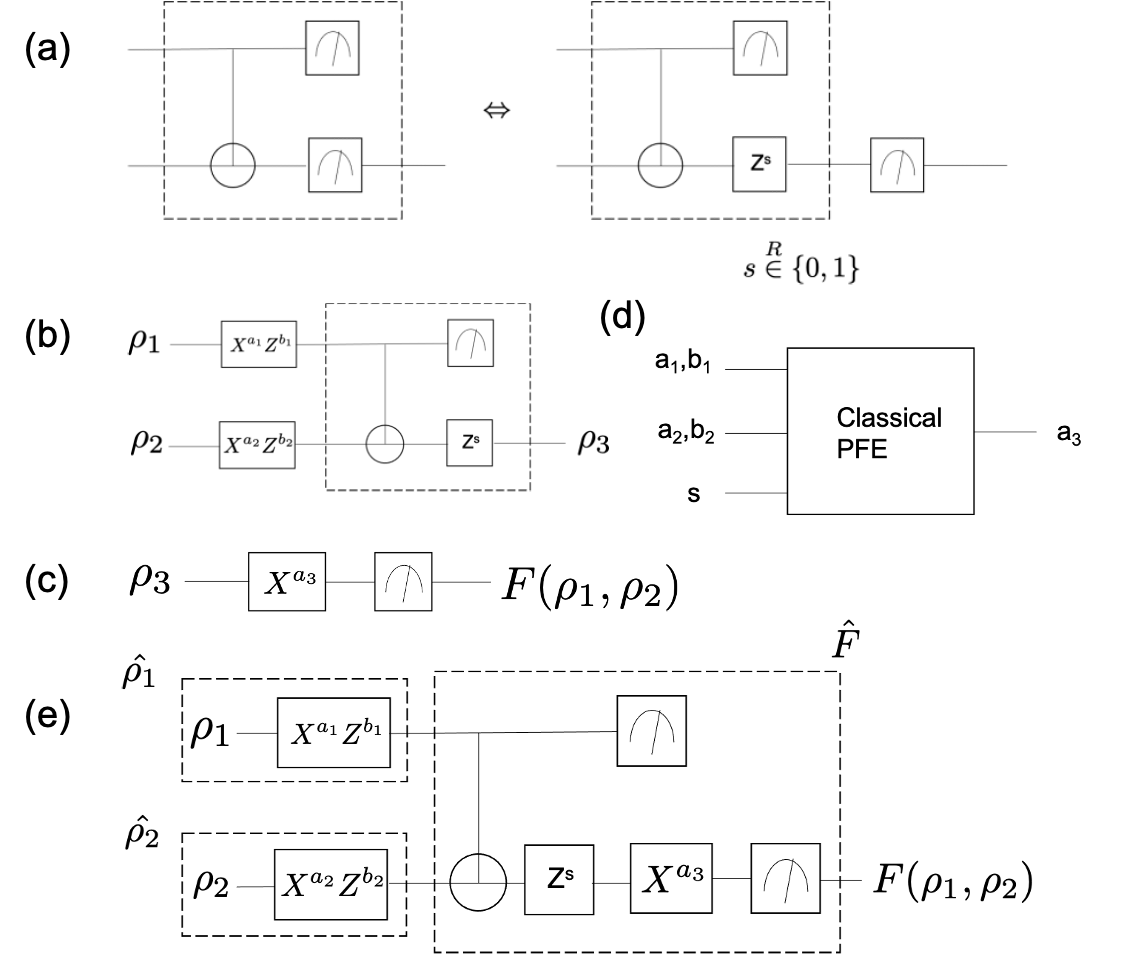}
\caption{(a)-(e) The design of the QPFE protocol for the example functionality. Details are explained in the main text.}
\label{fig:expscheme}
\end{figure}

In our experiment, we set $P_1$'s input to be either $\ket{0}$ or $\ket{1}$ and 
$P_2$'s input to be either $\ket{+}$ or $\ket{-}$. For this specific setting, we notice that the encoding turns the input qubit of the first party 
into either $\ket{0}$ or $\ket{1}$ and the input qubit of the second party into either $\ket{+}$
or $\ket{-}$. The additional $Z^s$ on the second qubit also leaves the second qubit as either  $\ket{+}$
or $\ket{-}$. Hence, for this specific setup, the outcome result $\rho_3$ is already correct without any correction of $X_{a_3}$,
namely $\rho_3 = F(\rho_1,\rho_2)$.
 
After explaining the protocol and the experiment setup, we now perform the experiment on the IBM quantum experience.
The quantum computer we use is ibmq\_lima. The connectivity of its qubits is shown in Fig.~\ref{fig:expresult}(a).
We will use two of the five qubits in the computer for our experimental demonstration. 
We first run the plain circuit $F$, which is depicted in Fig.~\ref{fig:expresult}(b). 
The result is shown in Fig.~\ref{fig:expresult}(c). It can be seen that the outcomes of 0 and 1 have approximately the same probability.
We then run the encrypted circuit depicted in Fig.~\ref{fig:expresult}(d) on the encoded input. 
The resulting probability distribution is shown in Fig.~\ref{fig:expresult}(e). 
 It can be seen that the outcomes of 0 and 1 are also approximately the same.
More precisely, 
\begin{equation}
 | \textrm{Pr} ( y =0 ) - \textrm{Pr}(y=1) | \le 0.012,
\end{equation}
where $y$ is short for $F(\rho_1, \rho_2)$. 
Since the experiment is repeated $N=1024$ times, the standard deviation of the probability
is $1/\sqrt{N} \approx 0.03$. Hence, the outcome distribution of the encrypted circuit
is within the standard deviation of the true distribution  $\frac{1}{2} \{ 0\} , \frac{1}{2} \{ 1\}$.

 \begin{figure}[htb]
\centering \includegraphics[width=9 cm]{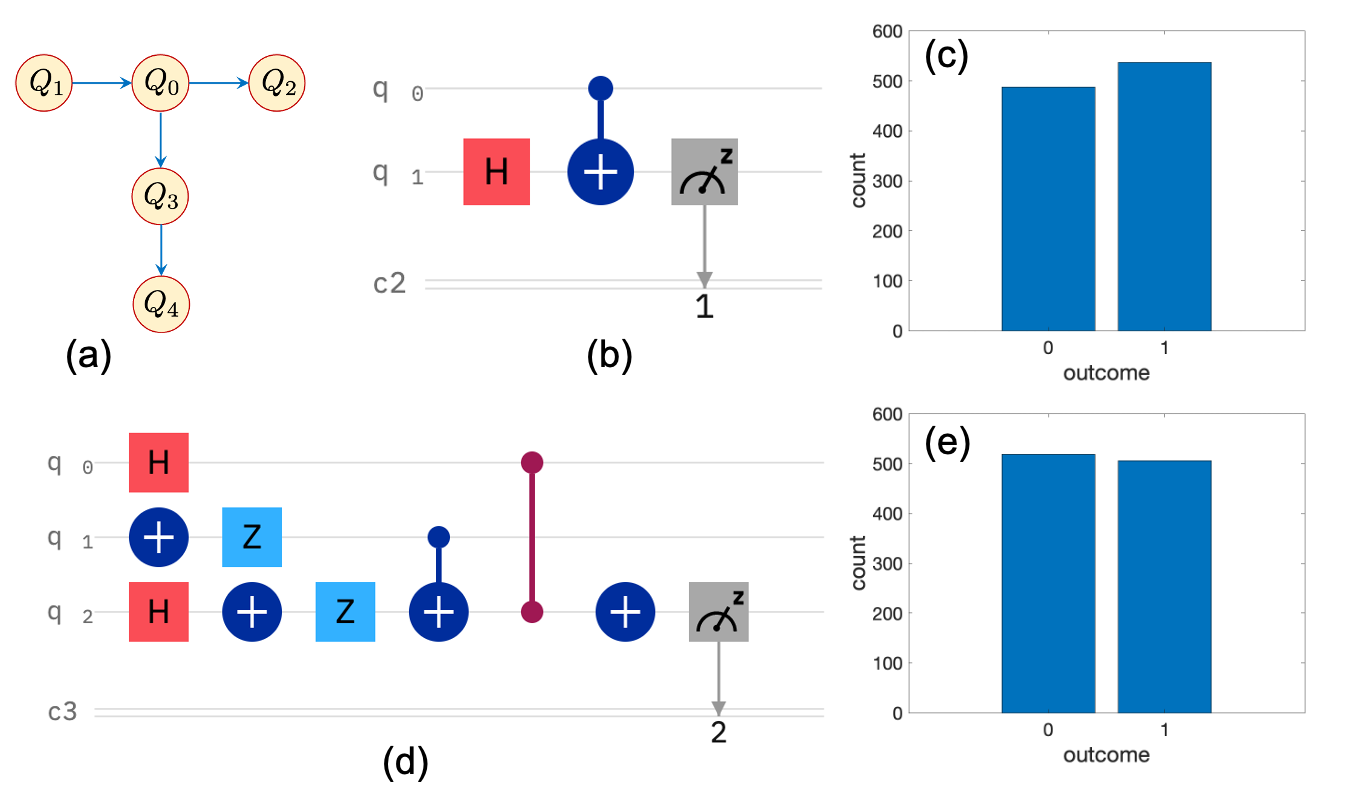}
\caption{The experimental setups and results. (a) The connectivity diagram of the quantum computer ibmq\_lima. 
(b) The plain circuit. 
(c) The outcome distribution of the plain circuit.
 (d) The encrypted circuit.
  (e) The outcome distribution of the encrypted circuit.}
\label{fig:expresult}
\end{figure}

A more faithful experimental demonstration of the QPFE scheme should set Alice and Bob at a distance,
and consequently would involve quantum communication between the distant parties
in addition to the quantum computation part. We leave this more faithful experimental demonstration of QPFE as future work.

\section{Application: quantum copy protection}
\label{sec:appl}

We now apply the QPFE scheme to the quantum copy protection (QCP) problem. 
A quantum copy protection problem considers the following scenario. A quantum software provider issues a quantum program $F$
to the user. The user may query the quantum program $F$ for many inputs $x_i$ and get the corresponding outputs $F(x_i)$. Here, the goal 
of the provider is that their software $F$ cannot be copied by the user so that other users will not be able to use the software unless they buy
another copy of the software from the provider.

The quantum copy protection problem is conceived in Ref.~\cite{aaronson2009quantum}. There, the author shows 
that the quantum copy protection is not possible for learnable programs. Indeed, if a program $F$ is learnable, namely it 
can be recovered from input-output pairs $\{ x_i , F(x_i)\}_{1\le i \le N}$, the user can simply reproduce the program by learning the  input-output pairs 
and then create arbitrary number copies of the program.  Later, it has been shown that even some unlearnable programs
also intrinsically do not allow any quantum copy protection schemes  \cite{ananth2021secure}. 

On the positive side, Aaronson \cite{aaronson2009quantum} has shown a quantum copy protection scheme based on
the quantum random oracle model for unlearnable functions. Later, Aaronson {\it et al.} \cite{aaronson2021new} weakened 
the assumption and has provided a quantum copy protection scheme based on
the classical random oracle model for unlearnable functions. Given the negative result above, the classical random oracle 
assumption is not always satisfied for unlearnable functions. Currently it is unknown for which class of functions  the classical random oracle 
assumption is satisfied.


Before we turn to the formal discussion of QCP, let us first review its definition in the literature \cite{aaronson2021new}. 
A standard QCP consists of three procedures:
\begin{itemize}
\item  $\textsf{Setup}(1^\lambda) \to (sk)$: Given a security parameter $\lambda$, the setup procedure generates a secret key $sk$ for the vendor.
\item  $\textsf{Generate}(sk, F) \to (\rho_F, \{ U_{F,x} \}_{x \in [N]})$: Using the secret key $sk$, the vendor encrypts the program $F$ into a quantum
state and a series of quantum unitaries $(\rho_F, \{ U_{F,x} \}_x) $ and issue them to the user.
\item  $\textsf{Compute}(   \rho_F, \{ U_{F,x} \}_{x \in [N]},x ) \to F(x)$: To compute the program $F$ on a specific input $x$, the user applies $U_{F,x}$
on the quantum state $\rho_F$ and measures the first qubit.
\end{itemize}
In this work, we relax this definition by only keeping its essentials, namely for any input $x$ that the user wishes to compute, the
vendor should return the result $F(x)$. We however do not restrict how the vendor achieves this goal, e.g., the vendor 
does not have to release $(\rho_F, \{ U_{F,x} \}_x) $ before user's query.
Next we define the security of QCP. The ideal functionality $\mathcal{F}_{QCP}$ of QCP takes input $x_i$ 
from the user and input $F$ from the vendor in the $i$th round. It then outputs $F(x_i)$ to the user.
A QCP protocol $\Pi_{QCP}$ is secure if for any independent and identically distributed (iid) adversary $\mathcal{A}$, there exists a simulator $\mathcal{S}$ that fulfills
\begin{equation}
\textrm{Exec}[\Pi_{QCP}, \mathcal{A}, \mathcal{Z}] \approx_c \textrm{Exec}[\mathcal{F}_{QCP}, \mathcal{S}, \mathcal{Z}].
\end{equation}

We now turn to our construction of the quantum copy protection scheme, which is based on
multi-runs of the QPFE scheme. In the $i$th run, Bob holds the private function $F$ and gets no output.
Alice provides a private quantum input $x_i$ and gets back a quantum output $F(x_i)$. 
After $N$ runs, Alice gets $N$ pairs $\{ x_i , F(x_i)\}_{1\le i \le N}$. The detailed process of this scheme is
shown in Protocol \ref{Fig:QCP}.

\begin{algorithm}
\caption{\textsc{QCP based on QPFE}}
\begin{flushleft}
\textbf{Alice Input:} $\{ x_i \}_{1\le i \le N}$.   

\textbf{Bob Input:} $F$. 

\textbf{Ingredient:} QPFE($x_A, x_B, y_A, y_B$), where $x_A$ and $x_B$ are Alice's and Bob's inputs, and 
$y_A$ and $y_B$ are Alice's and Bob's outputs.
\end{flushleft}
\begin{algorithmic}[1]
\FOR{$i$ = 1 to $n$}
 \STATE \quad QPFE($x_i$, $F$, $F(x_i)$, $\emptyset$).
\ENDFOR
\end{algorithmic}
\begin{flushleft}
\textbf{Alice Output:} $\{ F(x_i)\}_{1\le i \le N}$.
\end{flushleft}
\label{Fig:QCP}
\end{algorithm}

Now let us examine the security of this QCP scheme.
According to the security property of QPFE, in the $i$th run, we have
\begin{equation}
\textrm{Exec}[\Pi^i, \mathcal{A}^i, \mathcal{Z}] \approx_c \textrm{Exec}[\mathcal{F}^i, \mathcal{S}^i, \mathcal{Z}],
\end{equation}
where $\Pi^i$ is the QPFE protocol at the $i$-th run, $\mathcal{F}^i$ is the corresponding ideal functionality, and 
$\mathcal{S}^i$ is the corresponding simulator. Since any iid adversary $\mathcal{A}$ of QCP is a concatenation of $\mathcal{A}^i$,
the simulator $\mathcal{S}$ of QCP can be constructed
as a concatenation of $\mathcal{S}^i$, which results in
\begin{equation}
\textrm{Exec}[\Pi_{QCP}, \mathcal{A}, \mathcal{Z}] \approx_c \textrm{Exec}[\mathcal{F}_{QCP}, \mathcal{S}, \mathcal{Z}].
\end{equation}
This finishes the proof of the security of the QCP scheme. 

Note that this scheme fails for learnable functions, as one can still infer a function from multiple pairs of inputs and outputs
if the function is learnable.
However, for all unlearnable functions, it can well protect the function unlike previous schemes. 
This helps us enlarge the class of function that can be copy-protected.
Note that this scheme does not contradict with the no-go theorem for some unlearnable functions \cite{ananth2021secure}. In the no-go theorem,
it is assumed that the software provider has no interaction with the user beyond the initial software release stage,
while there are interactions between the software provider and the user in our construction of the quantum 
copy protection scheme based on QPFE. Hence, there is a tradeoff between the number of interactions and the
class of functions that can be copy protected. We leave the detailed investigation of this tradeoff as an interesting
research problem.

\section{Discussion}
\label{sec:conclusion}

In conclusion, we have initiated the study of QPFE.
We have formally defined the notion of QPFE as a quantum 
analogue of its classical counterpart. Then we presented two QPFE schemes together with their security proofs.
The first scheme makes black-box use of MPQC, while the second one builds on smaller 
primitives. The schemes presented can be applied to various quantum tasks where the quantum function
needs to be kept secret, such as quantum copy protection. An experimental demonstration of the scheme 
is also presented to facilitate its practical usage. 

There are a few interesting avenues for future investigation. First, it is an interesting 
direction to improve the computation and communication cost of the QPFE schemes
 that we present in this work. Second, it is interesting to explore what is 
 the minimum computational assumption that needs to be made to ensure the 
 existence of a QPFE scheme against a dishonest majority. Recently, there have been some works
 that aim to build quantum schemes with even weaker assumptions than the existence of a quantum one-way function \cite{morimae2022quantum}.
 Third, it is interesting to give a more faithful experimental demonstration of the schemes presented here.

\begin{acknowledgements}
This work was supported by the National Natural Science Foundation of China (Basic Science Center Program: 61988101), the National Natural Science Foundation of China (12105105), the Natural Science Foundation of Shanghai (21ZR1415800), the Shanghai Sailing Program (21YF1409800), and the startup fund from East China University of Science and Technology (YH0142214).
\end{acknowledgements}

\bibliographystyle{apsrev4-2}

\bibliography{BibliQPFE}

\begin{thebibliography}{29}%
\makeatletter
\providecommand \@ifxundefined [1]{%
 \@ifx{#1\undefined}
}%
\providecommand \@ifnum [1]{%
 \ifnum #1\expandafter \@firstoftwo
 \else \expandafter \@secondoftwo
 \fi
}%
\providecommand \@ifx [1]{%
 \ifx #1\expandafter \@firstoftwo
 \else \expandafter \@secondoftwo
 \fi
}%
\providecommand \natexlab [1]{#1}%
\providecommand \enquote  [1]{``#1''}%
\providecommand \bibnamefont  [1]{#1}%
\providecommand \bibfnamefont [1]{#1}%
\providecommand \citenamefont [1]{#1}%
\providecommand \href@noop [0]{\@secondoftwo}%
\providecommand \href [0]{\begingroup \@sanitize@url \@href}%
\providecommand \@href[1]{\@@startlink{#1}\@@href}%
\providecommand \@@href[1]{\endgroup#1\@@endlink}%
\providecommand \@sanitize@url [0]{\catcode `\\12\catcode `\$12\catcode
  `\&12\catcode `\#12\catcode `\^12\catcode `\_12\catcode `\%12\relax}%
\providecommand \@@startlink[1]{}%
\providecommand \@@endlink[0]{}%
\providecommand \url  [0]{\begingroup\@sanitize@url \@url }%
\providecommand \@url [1]{\endgroup\@href {#1}{\urlprefix }}%
\providecommand \urlprefix  [0]{URL }%
\providecommand \Eprint [0]{\href }%
\providecommand \doibase [0]{https://doi.org/}%
\providecommand \selectlanguage [0]{\@gobble}%
\providecommand \bibinfo  [0]{\@secondoftwo}%
\providecommand \bibfield  [0]{\@secondoftwo}%
\providecommand \translation [1]{[#1]}%
\providecommand \BibitemOpen [0]{}%
\providecommand \bibitemStop [0]{}%
\providecommand \bibitemNoStop [0]{.\EOS\space}%
\providecommand \EOS [0]{\spacefactor3000\relax}%
\providecommand \BibitemShut  [1]{\csname bibitem#1\endcsname}%
\let\auto@bib@innerbib\@empty
\bibitem [{\citenamefont {Abadi}\ and\ \citenamefont
  {Feigenbaum}(1990)}]{abadi1990secure}%
  \BibitemOpen
  \bibfield  {author} {\bibinfo {author} {\bibfnamefont {M.}~\bibnamefont
  {Abadi}}\ and\ \bibinfo {author} {\bibfnamefont {J.}~\bibnamefont
  {Feigenbaum}},\ }\href@noop {} {\bibfield  {journal} {\bibinfo  {journal}
  {Journal of Cryptology}\ }\textbf {\bibinfo {volume} {2}},\ \bibinfo {pages}
  {1} (\bibinfo {year} {1990})}\BibitemShut {NoStop}%
\bibitem [{\citenamefont {Mohassel}\ \emph {et~al.}(2014)\citenamefont
  {Mohassel}, \citenamefont {Sadeghian},\ and\ \citenamefont
  {Smart}}]{mohassel2014actively}%
  \BibitemOpen
  \bibfield  {author} {\bibinfo {author} {\bibfnamefont {P.}~\bibnamefont
  {Mohassel}}, \bibinfo {author} {\bibfnamefont {S.}~\bibnamefont
  {Sadeghian}},\ and\ \bibinfo {author} {\bibfnamefont {N.~P.}\ \bibnamefont
  {Smart}},\ }in\ \href@noop {} {\emph {\bibinfo {booktitle} {International
  Conference on the Theory and Application of Cryptology and Information
  Security}}}\ (\bibinfo {organization} {Springer},\ \bibinfo {year} {2014})\
  pp.\ \bibinfo {pages} {486--505}\BibitemShut {NoStop}%
\bibitem [{\citenamefont {Bartusek}\ \emph
  {et~al.}(2021{\natexlab{a}})\citenamefont {Bartusek}, \citenamefont
  {Coladangelo}, \citenamefont {Khurana},\ and\ \citenamefont
  {Ma}}]{bartusek2021one}%
  \BibitemOpen
  \bibfield  {author} {\bibinfo {author} {\bibfnamefont {J.}~\bibnamefont
  {Bartusek}}, \bibinfo {author} {\bibfnamefont {A.}~\bibnamefont
  {Coladangelo}}, \bibinfo {author} {\bibfnamefont {D.}~\bibnamefont
  {Khurana}},\ and\ \bibinfo {author} {\bibfnamefont {F.}~\bibnamefont {Ma}},\
  }in\ \href@noop {} {\emph {\bibinfo {booktitle} {Annual International
  Cryptology Conference}}}\ (\bibinfo {organization} {Springer},\ \bibinfo
  {year} {2021})\ pp.\ \bibinfo {pages} {467--496}\BibitemShut {NoStop}%
\bibitem [{\citenamefont {Kiss}\ and\ \citenamefont
  {Schneider}(2016)}]{kiss2016valiant}%
  \BibitemOpen
  \bibfield  {author} {\bibinfo {author} {\bibfnamefont {{\'A}.}~\bibnamefont
  {Kiss}}\ and\ \bibinfo {author} {\bibfnamefont {T.}~\bibnamefont
  {Schneider}},\ }in\ \href@noop {} {\emph {\bibinfo {booktitle} {Annual
  International Conference on the Theory and Applications of Cryptographic
  Techniques}}}\ (\bibinfo {organization} {Springer},\ \bibinfo {year} {2016})\
  pp.\ \bibinfo {pages} {699--728}\BibitemShut {NoStop}%
\bibitem [{\citenamefont {G{\"u}nther}\ \emph {et~al.}(2017)\citenamefont
  {G{\"u}nther}, \citenamefont {Kiss},\ and\ \citenamefont
  {Schneider}}]{gunther2017more}%
  \BibitemOpen
  \bibfield  {author} {\bibinfo {author} {\bibfnamefont {D.}~\bibnamefont
  {G{\"u}nther}}, \bibinfo {author} {\bibfnamefont {{\'A}.}~\bibnamefont
  {Kiss}},\ and\ \bibinfo {author} {\bibfnamefont {T.}~\bibnamefont
  {Schneider}},\ }in\ \href@noop {} {\emph {\bibinfo {booktitle} {International
  Conference on the Theory and Application of Cryptology and Information
  Security}}}\ (\bibinfo {organization} {Springer},\ \bibinfo {year} {2017})\
  pp.\ \bibinfo {pages} {443--470}\BibitemShut {NoStop}%
\bibitem [{\citenamefont {Bing{\"o}l}\ \emph {et~al.}(2019)\citenamefont
  {Bing{\"o}l}, \citenamefont {Bi{\c{c}}er}, \citenamefont {Kiraz},\ and\
  \citenamefont {Levi}}]{bingol2019efficient}%
  \BibitemOpen
  \bibfield  {author} {\bibinfo {author} {\bibfnamefont {M.~A.}\ \bibnamefont
  {Bing{\"o}l}}, \bibinfo {author} {\bibfnamefont {O.}~\bibnamefont
  {Bi{\c{c}}er}}, \bibinfo {author} {\bibfnamefont {M.~S.}\ \bibnamefont
  {Kiraz}},\ and\ \bibinfo {author} {\bibfnamefont {A.}~\bibnamefont {Levi}},\
  }\href@noop {} {\bibfield  {journal} {\bibinfo  {journal} {The Computer
  Journal}\ }\textbf {\bibinfo {volume} {62}},\ \bibinfo {pages} {598}
  (\bibinfo {year} {2019})}\BibitemShut {NoStop}%
\bibitem [{\citenamefont {Bi{\c{c}}er}\ \emph {et~al.}(2020)\citenamefont
  {Bi{\c{c}}er}, \citenamefont {Bingol}, \citenamefont {Kiraz},\ and\
  \citenamefont {Levi}}]{biccer2020highly}%
  \BibitemOpen
  \bibfield  {author} {\bibinfo {author} {\bibfnamefont {O.}~\bibnamefont
  {Bi{\c{c}}er}}, \bibinfo {author} {\bibfnamefont {M.~A.}\ \bibnamefont
  {Bingol}}, \bibinfo {author} {\bibfnamefont {M.~S.}\ \bibnamefont {Kiraz}},\
  and\ \bibinfo {author} {\bibfnamefont {A.}~\bibnamefont {Levi}},\ }\href@noop
  {} {\bibfield  {journal} {\bibinfo  {journal} {IEEE Transactions on
  Dependable and Secure Computing}\ } (\bibinfo {year} {2020})}\BibitemShut
  {NoStop}%
\bibitem [{\citenamefont {Cr{\'e}peau}\ \emph {et~al.}(2002)\citenamefont
  {Cr{\'e}peau}, \citenamefont {Gottesman},\ and\ \citenamefont
  {Smith}}]{crepeau2002secure}%
  \BibitemOpen
  \bibfield  {author} {\bibinfo {author} {\bibfnamefont {C.}~\bibnamefont
  {Cr{\'e}peau}}, \bibinfo {author} {\bibfnamefont {D.}~\bibnamefont
  {Gottesman}},\ and\ \bibinfo {author} {\bibfnamefont {A.}~\bibnamefont
  {Smith}},\ }in\ \href@noop {} {\emph {\bibinfo {booktitle} {Proceedings of
  the thiry-fourth annual ACM symposium on Theory of computing}}}\ (\bibinfo
  {year} {2002})\ pp.\ \bibinfo {pages} {643--652}\BibitemShut {NoStop}%
\bibitem [{\citenamefont {Bartusek}\ \emph
  {et~al.}(2021{\natexlab{b}})\citenamefont {Bartusek}, \citenamefont
  {Coladangelo}, \citenamefont {Khurana},\ and\ \citenamefont
  {Ma}}]{bartusek2021round}%
  \BibitemOpen
  \bibfield  {author} {\bibinfo {author} {\bibfnamefont {J.}~\bibnamefont
  {Bartusek}}, \bibinfo {author} {\bibfnamefont {A.}~\bibnamefont
  {Coladangelo}}, \bibinfo {author} {\bibfnamefont {D.}~\bibnamefont
  {Khurana}},\ and\ \bibinfo {author} {\bibfnamefont {F.}~\bibnamefont {Ma}},\
  }in\ \href@noop {} {\emph {\bibinfo {booktitle} {Annual International
  Cryptology Conference}}}\ (\bibinfo {organization} {Springer},\ \bibinfo
  {year} {2021})\ pp.\ \bibinfo {pages} {406--435}\BibitemShut {NoStop}%
\bibitem [{\citenamefont {Ben-Or}\ \emph {et~al.}(2006)\citenamefont {Ben-Or},
  \citenamefont {Cr{\'e}peau}, \citenamefont {Gottesman}, \citenamefont
  {Hassidim},\ and\ \citenamefont {Smith}}]{ben2006secure}%
  \BibitemOpen
  \bibfield  {author} {\bibinfo {author} {\bibfnamefont {M.}~\bibnamefont
  {Ben-Or}}, \bibinfo {author} {\bibfnamefont {C.}~\bibnamefont {Cr{\'e}peau}},
  \bibinfo {author} {\bibfnamefont {D.}~\bibnamefont {Gottesman}}, \bibinfo
  {author} {\bibfnamefont {A.}~\bibnamefont {Hassidim}},\ and\ \bibinfo
  {author} {\bibfnamefont {A.}~\bibnamefont {Smith}},\ }in\ \href@noop {}
  {\emph {\bibinfo {booktitle} {2006 47th Annual IEEE Symposium on Foundations
  of Computer Science (FOCS'06)}}}\ (\bibinfo {organization} {IEEE},\ \bibinfo
  {year} {2006})\ pp.\ \bibinfo {pages} {249--260}\BibitemShut {NoStop}%
\bibitem [{\citenamefont {Dulek}\ \emph {et~al.}(2020)\citenamefont {Dulek},
  \citenamefont {Grilo}, \citenamefont {Jeffery}, \citenamefont {Majenz},\ and\
  \citenamefont {Schaffner}}]{dulek2020secure}%
  \BibitemOpen
  \bibfield  {author} {\bibinfo {author} {\bibfnamefont {Y.}~\bibnamefont
  {Dulek}}, \bibinfo {author} {\bibfnamefont {A.~B.}\ \bibnamefont {Grilo}},
  \bibinfo {author} {\bibfnamefont {S.}~\bibnamefont {Jeffery}}, \bibinfo
  {author} {\bibfnamefont {C.}~\bibnamefont {Majenz}},\ and\ \bibinfo {author}
  {\bibfnamefont {C.}~\bibnamefont {Schaffner}},\ }in\ \href@noop {} {\emph
  {\bibinfo {booktitle} {Annual International Conference on the Theory and
  Applications of Cryptographic Techniques}}}\ (\bibinfo {organization}
  {Springer},\ \bibinfo {year} {2020})\ pp.\ \bibinfo {pages}
  {729--758}\BibitemShut {NoStop}%
\bibitem [{\citenamefont {Yu}\ \emph {et~al.}(2014)\citenamefont {Yu},
  \citenamefont {P{\'e}rez-Delgado},\ and\ \citenamefont
  {Fitzsimons}}]{yu2014limitations}%
  \BibitemOpen
  \bibfield  {author} {\bibinfo {author} {\bibfnamefont {L.}~\bibnamefont
  {Yu}}, \bibinfo {author} {\bibfnamefont {C.~A.}\ \bibnamefont
  {P{\'e}rez-Delgado}},\ and\ \bibinfo {author} {\bibfnamefont {J.~F.}\
  \bibnamefont {Fitzsimons}},\ }\href@noop {} {\bibfield  {journal} {\bibinfo
  {journal} {Physical Review A}\ }\textbf {\bibinfo {volume} {90}},\ \bibinfo
  {pages} {050303} (\bibinfo {year} {2014})}\BibitemShut {NoStop}%
\bibitem [{\citenamefont {Broadbent}\ and\ \citenamefont
  {Jeffery}(2015)}]{broadbent2015quantum}%
  \BibitemOpen
  \bibfield  {author} {\bibinfo {author} {\bibfnamefont {A.}~\bibnamefont
  {Broadbent}}\ and\ \bibinfo {author} {\bibfnamefont {S.}~\bibnamefont
  {Jeffery}},\ }in\ \href@noop {} {\emph {\bibinfo {booktitle} {Annual
  Cryptology Conference}}}\ (\bibinfo {organization} {Springer},\ \bibinfo
  {year} {2015})\ pp.\ \bibinfo {pages} {609--629}\BibitemShut {NoStop}%
\bibitem [{\citenamefont {Dulek}\ \emph {et~al.}(2016)\citenamefont {Dulek},
  \citenamefont {Schaffner},\ and\ \citenamefont
  {Speelman}}]{dulek2016quantum}%
  \BibitemOpen
  \bibfield  {author} {\bibinfo {author} {\bibfnamefont {Y.}~\bibnamefont
  {Dulek}}, \bibinfo {author} {\bibfnamefont {C.}~\bibnamefont {Schaffner}},\
  and\ \bibinfo {author} {\bibfnamefont {F.}~\bibnamefont {Speelman}},\ }in\
  \href@noop {} {\emph {\bibinfo {booktitle} {Annual International Cryptology
  Conference}}}\ (\bibinfo {organization} {Springer},\ \bibinfo {year} {2016})\
  pp.\ \bibinfo {pages} {3--32}\BibitemShut {NoStop}%
\bibitem [{\citenamefont {Alagic}\ \emph {et~al.}(2017)\citenamefont {Alagic},
  \citenamefont {Dulek}, \citenamefont {Schaffner},\ and\ \citenamefont
  {Speelman}}]{alagic2017quantum}%
  \BibitemOpen
  \bibfield  {author} {\bibinfo {author} {\bibfnamefont {G.}~\bibnamefont
  {Alagic}}, \bibinfo {author} {\bibfnamefont {Y.}~\bibnamefont {Dulek}},
  \bibinfo {author} {\bibfnamefont {C.}~\bibnamefont {Schaffner}},\ and\
  \bibinfo {author} {\bibfnamefont {F.}~\bibnamefont {Speelman}},\ }in\
  \href@noop {} {\emph {\bibinfo {booktitle} {International Conference on the
  Theory and Application of Cryptology and Information Security}}}\ (\bibinfo
  {organization} {Springer},\ \bibinfo {year} {2017})\ pp.\ \bibinfo {pages}
  {438--467}\BibitemShut {NoStop}%
\bibitem [{\citenamefont {Mahadev}(2018)}]{mahadev2018classical2}%
  \BibitemOpen
  \bibfield  {author} {\bibinfo {author} {\bibfnamefont {U.}~\bibnamefont
  {Mahadev}},\ }in\ \href@noop {} {\emph {\bibinfo {booktitle} {2018 IEEE 59th
  Annual Symposium on Foundations of Computer Science (FOCS)}}}\ (\bibinfo
  {organization} {IEEE Computer Society},\ \bibinfo {year} {2018})\ pp.\
  \bibinfo {pages} {332--338}\BibitemShut {NoStop}%
\bibitem [{\citenamefont {Brakerski}(2018)}]{brakerski2018quantum}%
  \BibitemOpen
  \bibfield  {author} {\bibinfo {author} {\bibfnamefont {Z.}~\bibnamefont
  {Brakerski}},\ }in\ \href@noop {} {\emph {\bibinfo {booktitle} {Annual
  International Cryptology Conference}}}\ (\bibinfo {organization} {Springer},\
  \bibinfo {year} {2018})\ pp.\ \bibinfo {pages} {67--95}\BibitemShut {NoStop}%
\bibitem [{\citenamefont {Ouyang}\ \emph {et~al.}(2018)\citenamefont {Ouyang},
  \citenamefont {Tan},\ and\ \citenamefont {Fitzsimons}}]{ouyang2018quantum}%
  \BibitemOpen
  \bibfield  {author} {\bibinfo {author} {\bibfnamefont {Y.}~\bibnamefont
  {Ouyang}}, \bibinfo {author} {\bibfnamefont {S.-H.}\ \bibnamefont {Tan}},\
  and\ \bibinfo {author} {\bibfnamefont {J.~F.}\ \bibnamefont {Fitzsimons}},\
  }\href@noop {} {\bibfield  {journal} {\bibinfo  {journal} {Physical Review
  A}\ }\textbf {\bibinfo {volume} {98}},\ \bibinfo {pages} {042334} (\bibinfo
  {year} {2018})}\BibitemShut {NoStop}%
\bibitem [{\citenamefont {Canetti}(2001)}]{canetti2001universally}%
  \BibitemOpen
  \bibfield  {author} {\bibinfo {author} {\bibfnamefont {R.}~\bibnamefont
  {Canetti}},\ }in\ \href@noop {} {\emph {\bibinfo {booktitle} {Proceedings
  42nd IEEE Symposium on Foundations of Computer Science}}}\ (\bibinfo
  {organization} {IEEE},\ \bibinfo {year} {2001})\ pp.\ \bibinfo {pages}
  {136--145}\BibitemShut {NoStop}%
\bibitem [{\citenamefont {Ishai}\ \emph {et~al.}(2008)\citenamefont {Ishai},
  \citenamefont {Prabhakaran},\ and\ \citenamefont
  {Sahai}}]{Ishai2008founding}%
  \BibitemOpen
  \bibfield  {author} {\bibinfo {author} {\bibfnamefont {Y.}~\bibnamefont
  {Ishai}}, \bibinfo {author} {\bibfnamefont {M.}~\bibnamefont {Prabhakaran}},\
  and\ \bibinfo {author} {\bibfnamefont {A.}~\bibnamefont {Sahai}},\ }in\
  \href@noop {} {\emph {\bibinfo {booktitle} {Advances in Cryptology -- CRYPTO
  2008}}},\ \bibinfo {editor} {edited by\ \bibinfo {editor} {\bibfnamefont
  {D.}~\bibnamefont {Wagner}}}\ (\bibinfo  {publisher} {Springer Berlin
  Heidelberg},\ \bibinfo {address} {Berlin, Heidelberg},\ \bibinfo {year}
  {2008})\ pp.\ \bibinfo {pages} {572--591}\BibitemShut {NoStop}%
\bibitem [{\citenamefont {Brakerski}\ and\ \citenamefont
  {Yuen}(2022)}]{brakerski2022quantum}%
  \BibitemOpen
  \bibfield  {author} {\bibinfo {author} {\bibfnamefont {Z.}~\bibnamefont
  {Brakerski}}\ and\ \bibinfo {author} {\bibfnamefont {H.}~\bibnamefont
  {Yuen}},\ }in\ \href@noop {} {\emph {\bibinfo {booktitle} {Proceedings of the
  54th Annual ACM SIGACT Symposium on Theory of Computing}}}\ (\bibinfo {year}
  {2022})\ pp.\ \bibinfo {pages} {804--817}\BibitemShut {NoStop}%
\bibitem [{\citenamefont {Liu}\ \emph {et~al.}(2022)\citenamefont {Liu},
  \citenamefont {Wang},\ and\ \citenamefont {Yiu}}]{liu2022making}%
  \BibitemOpen
  \bibfield  {author} {\bibinfo {author} {\bibfnamefont {Y.}~\bibnamefont
  {Liu}}, \bibinfo {author} {\bibfnamefont {Q.}~\bibnamefont {Wang}},\ and\
  \bibinfo {author} {\bibfnamefont {S.-M.}\ \bibnamefont {Yiu}},\ }in\
  \href@noop {} {\emph {\bibinfo {booktitle} {Proceedings of the 25th IACR
  International Conference on Practice and Theory of Public-Key
  Cryptography}}}\ (\bibinfo {organization} {Springer},\ \bibinfo {year}
  {2022})\ pp.\ \bibinfo {pages} {349--378}\BibitemShut {NoStop}%
\bibitem [{\citenamefont {Jamio{\l}kowski}(1972)}]{jamiolkowski1972linear}%
  \BibitemOpen
  \bibfield  {author} {\bibinfo {author} {\bibfnamefont {A.}~\bibnamefont
  {Jamio{\l}kowski}},\ }\href@noop {} {\bibfield  {journal} {\bibinfo
  {journal} {Reports on Mathematical Physics}\ }\textbf {\bibinfo {volume}
  {3}},\ \bibinfo {pages} {275} (\bibinfo {year} {1972})}\BibitemShut {NoStop}%
\bibitem [{\citenamefont {Choi}(1975)}]{choi1975completely}%
  \BibitemOpen
  \bibfield  {author} {\bibinfo {author} {\bibfnamefont {M.-D.}\ \bibnamefont
  {Choi}},\ }\href@noop {} {\bibfield  {journal} {\bibinfo  {journal} {Linear
  algebra and its applications}\ }\textbf {\bibinfo {volume} {10}},\ \bibinfo
  {pages} {285} (\bibinfo {year} {1975})}\BibitemShut {NoStop}%
\bibitem [{ibm()}]{ibmQ}%
  \BibitemOpen
  \href@noop {} {\bibinfo  {journal} {https://quantum-computing.ibm.com}\
  }\BibitemShut {NoStop}%
\bibitem [{\citenamefont {Aaronson}(2009)}]{aaronson2009quantum}%
  \BibitemOpen
\bibfield  {journal} {  }\bibfield  {author} {\bibinfo {author} {\bibfnamefont
  {S.}~\bibnamefont {Aaronson}},\ }in\ \href@noop {} {\emph {\bibinfo
  {booktitle} {24th Annual IEEE Conference on Computational Complexity}}}\
  (\bibinfo {organization} {IEEE},\ \bibinfo {year} {2009})\ pp.\ \bibinfo
  {pages} {229--242}\BibitemShut {NoStop}%
\bibitem [{\citenamefont {Ananth}\ and\ \citenamefont
  {Placa}(2021)}]{ananth2021secure}%
  \BibitemOpen
  \bibfield  {author} {\bibinfo {author} {\bibfnamefont {P.}~\bibnamefont
  {Ananth}}\ and\ \bibinfo {author} {\bibfnamefont {R.~L.~L.}\ \bibnamefont
  {Placa}},\ }in\ \href@noop {} {\emph {\bibinfo {booktitle} {Annual
  International Conference on the Theory and Applications of Cryptographic
  Techniques}}}\ (\bibinfo {organization} {Springer},\ \bibinfo {year} {2021})\
  pp.\ \bibinfo {pages} {501--530}\BibitemShut {NoStop}%
\bibitem [{\citenamefont {Aaronson}\ \emph {et~al.}(2021)\citenamefont
  {Aaronson}, \citenamefont {Liu}, \citenamefont {Liu}, \citenamefont
  {Zhandry},\ and\ \citenamefont {Zhang}}]{aaronson2021new}%
  \BibitemOpen
  \bibfield  {author} {\bibinfo {author} {\bibfnamefont {S.}~\bibnamefont
  {Aaronson}}, \bibinfo {author} {\bibfnamefont {J.}~\bibnamefont {Liu}},
  \bibinfo {author} {\bibfnamefont {Q.}~\bibnamefont {Liu}}, \bibinfo {author}
  {\bibfnamefont {M.}~\bibnamefont {Zhandry}},\ and\ \bibinfo {author}
  {\bibfnamefont {R.}~\bibnamefont {Zhang}},\ }in\ \href@noop {} {\emph
  {\bibinfo {booktitle} {Advances in Cryptology -- CRYPTO 2021}}},\ \bibinfo
  {editor} {edited by\ \bibinfo {editor} {\bibfnamefont {T.}~\bibnamefont
  {Malkin}}\ and\ \bibinfo {editor} {\bibfnamefont {C.}~\bibnamefont
  {Peikert}}}\ (\bibinfo  {publisher} {Springer International Publishing},\
  \bibinfo {address} {Cham},\ \bibinfo {year} {2021})\ pp.\ \bibinfo {pages}
  {526--555}\BibitemShut {NoStop}%
\bibitem [{\citenamefont {Morimae}\ and\ \citenamefont
  {Yamakawa}(2022)}]{morimae2022quantum}%
  \BibitemOpen
  \bibfield  {author} {\bibinfo {author} {\bibfnamefont {T.}~\bibnamefont
  {Morimae}}\ and\ \bibinfo {author} {\bibfnamefont {T.}~\bibnamefont
  {Yamakawa}},\ }in\ \href@noop {} {\emph {\bibinfo {booktitle} {Annual
  International Cryptology Conference}}}\ (\bibinfo {organization} {Springer},\
  \bibinfo {year} {2022})\ pp.\ \bibinfo {pages} {269--295}\BibitemShut
  {NoStop}%
\end{thebibliography}%

\end{document}